\documentclass[letterpaper,11pt]{article}
\usepackage[margin=1in]{geometry}
\usepackage{graphicx}
\usepackage{amsfonts,amsmath,amsthm,thmtools}%
\usepackage{mathtools}
\usepackage{bm}
\usepackage{enumitem}
\usepackage[dvipsnames]{xcolor}
\usepackage{hyperref}
\usepackage{comment}
\usepackage{float}
\usepackage{cleveref}
\usepackage{algorithm}
\usepackage{algorithmic}
\usepackage{subcaption}
\usepackage{xspace}
\usepackage{tcolorbox}
\usepackage{nicefrac,xfrac}

\usepackage{dsfont}
\usepackage[square,numbers,sort]{natbib}
\usepackage{tikz}

\declaretheorem[style=definition]{definition}

\declaretheorem[style=definition]{example}
\declaretheorem[style=definition]{remark}

\declaretheorem{theorem}
\declaretheorem[sibling=theorem]{lemma}
\declaretheorem[sibling=theorem]{corollary}

\declaretheorem[sibling=theorem]{proposition}

\newcommand{\Xcomment}[1]{{}}

\newcommand{\dd}{\,\mathrm{d}}
\newcommand{\bigO}{\mathcal{O}}

\DeclareMathOperator{\E}{\mathbf{E}}
\let\Pr\PrOld
\DeclareMathOperator{\Pr}{\mathbf{Pr}}

\newcommand{\W}{\mathcal{W}} %
\newcommand{\Wone}{\ee^{\W_{-1}(-\beta)}}
\newcommand{\Wzero}{\ee^{\W_{0}(-\beta)}}

\newcommand{\MaxProb}{\textrm{MaxProb}\xspace}
\newcommand{\MaxExp}{\textrm{MaxExp}\xspace}

\newcommand{\ee}{\mathrm{e}}

\newcommand{\AlgVal}{\textsc{Alg}}
\newcommand{\ALG}{x_{\AlgVal}}
\newcommand{\OPT}{x_{\textsc{Opt}}}

\newcommand{\prior}{\tilde{F}}

\newcommand{\stoptime}{\AlgVal}
\newcommand{\threshold}{\theta}
\newcommand{\quantile}{q}
\newcommand{\maxexpoptthres}{\threshold_{n}^{\scriptscriptstyle \mathrm{exp}}}
\newcommand{\maxproboptthres}{\threshold_{n}^{\scriptscriptstyle \mathrm{prob}}}
\newcommand{\constint}{\gamma} %
\newcommand{\restrict}[1]{#1\big|_{\scriptscriptstyle \lambda_{1},\lambda_{2}}}
\newcommand{\restrictsmall}[1]{#1|_{\scriptscriptstyle \lambda_{1},\lambda_{2}}}

\newcommand{\algowiththres}[0]{\mathcal{A}}

\DeclareMathOperator*{\argmax}{arg\,max}
\newcommand{\Ind}[1]{\mathds{1}{\scriptstyle{\{#1\}}}}

\newcommand{\accept}{\textsc{Acc}}
\newcommand{\rejall}{\textsc{Rej}}

\newcommand{\N}{\mathbb{N}}

\title{Optimal Stopping with a Predicted Prior\footnote{This second arXiv version improves the exposition based on feedback we received after posting the first version. We acknowledge the concurrent work by Kehne and Kesselheim~\cite{kehne2025prophet}, add \Cref{subsec:optimal-robust-thresholds} that explains how to use the indifference Bellman equation to compute the optimal threshold function within our algorithm class, and strengthen~\Cref{thm:max-exp-consistency-robustness} by establishing that the trade‑off curve for \MaxProb is tight within the class.
}}
\author{
Tian Bai\footnote{University of Bergen. Email: tian.bai@uib.no. This work was done when the author was 
at the University of Hong Kong.}
\and
Zhiyi Huang\footnote{The University of Hong Kong. Email: zhiyi@cs.hku.hk, dongchen.li@connect.hku.hk}
\and
Chui Shan Lee\footnote{Hong Kong University of Science and Technology. Email: tracylcs@ust.hk. This work was done when the author was 
at the University of Hong Kong.}
\and
Dongchen Li\footnotemark[2]
}
\date{\today}

\begin{document}

\begin{titlepage}
\thispagestyle{empty}
\maketitle
\begin{abstract}
There are two major models of value uncertainty in the optimal stopping literature:
the secretary model, which assumes no prior knowledge, and the prophet inequality model, which assumes full information about value distributions.
In practice,
decision makers often rely on machine-learned priors that may be erroneous.
Motivated by this gap, we formulate the model of \emph{optimal stopping with a predicted prior} to design algorithms that are both consistent, exploiting the prediction when accurate, and robust, retaining worst-case guarantees when it is not.

Existing secretary and prophet inequality algorithms are either pessimistic in consistency or not robust to misprediction.
A randomized combination only interpolates their guarantees linearly.
We show that a family of bi-criteria algorithms achieves improved consistency-robustness trade-offs, both for maximizing the expected accepted value and for maximizing the probability of accepting the maximum value.
We further prove that for the latter objective, no algorithm can simultaneously match the best prophet inequality algorithm in consistency, and the best secretary algorithm in robustness.

\end{abstract}
\end{titlepage}

\section{Introduction}

Optimal stopping problems study a fundamental dilemma in decision‑making: to seize an expiring opportunity now, or to wait for a better chance that may never come.
A classic example is the \emph{secretary problem}.
In this problem, $n$ values are revealed to a decision maker in a random order.  
After observing each value, the decision maker immediately decides whether to accept the value.
The decision maker can accept at most one value and aims to maximize its expectation.
An optimal solution is \emph{Dynkin's algorithm} \cite{dynkinOptimumChoiceInstant1963,fergusonWhoSolvedSecretary1989}, which rejects the first $\lfloor\nicefrac{n}{\ee}\rfloor$ values, and then accepts the first value that exceeds all previous ones. 
It guarantees that the expected accepted value is at least $\nicefrac{1}{\ee}$ times the maximum.
This ratio of the algorithm's expected accepted value to the maximum value in hindsight is called the \emph{competitive ratio}.

In sharp contrast to the minimal prior knowledge assumed in the secretary problem, \emph{prophet inequality} considers values drawn independently from \emph{fully known} prior distributions.
In the i.i.d.\ case where values are drawn from a common distribution, the optimal algorithm guarantees that the expected accepted value is at least approximately $0.745$ times the expected maximum value~\cite{hillComparisonsStopRule1982,kertzStopRuleSupremum1986,correaPostedPriceMechanisms2017,liuVariableDecompositionProphet2021},
substantially better than the prior-free secretary guarantee of $\nicefrac{1}{\ee} \approx 0.368$.

On the one hand, the secretary problem does not capture the abundant data available to decision makers in this era.
On the other hand, the assumption of complete knowledge of the prior in prophet inequality is idealized. 
What real-world decision makers often have are machine-learned and potentially erroneous \emph{predicted priors}. 
\citet{duttingPostedPricingProphet2019} took a step in this direction.
They show that the performance of the optimal backward-induction algorithm that \emph{fully trusts} the predicted prior degrades proportionally to some appropriate measures of prediction error, thereby providing a smoothness-style guarantee in the framework of \emph{algorithms with predictions}~\cite{mitzenmacher2022algorithms}.

\subsection{Optimal Stopping with a Predicted Prior}

We formulate the problem of \emph{optimal stopping 
with a predicted prior} and focus on the trade-off between  \emph{consistency} and \emph{robustness}, i.e., the performance when the predicted prior is correct and when it may be arbitrarily wrong.
Consider the optimal stopping problem with $n$ values drawn i.i.d.\ from an \emph{unknown} prior distribution $F$.
The algorithm is given the number of values $n$ and a \emph{predicted prior} $\tilde{F}$ in the beginning.
Then, the algorithm decides whether to accept each value upon its arrival, based on the history of observed values and the predicted prior $\tilde{F}$.
We consider two standard objectives: 1) maximizing the expected accepted value (\MaxExp), which we have already discussed, and 2) maximizing the probability of accepting the maximum value (\MaxProb).

An algorithm is \emph{$\alpha$-consistent} if it is $\alpha$-competitive when the predicted prior is correct.
It is \emph{$\beta$-robust} if it remains $\beta$-competitive even when the (mis)predicted prior is arbitrary.
Ideally, we want an algorithm to be both $\alpha$-consistent and $\beta$-robust with $\alpha$ and $\beta$ being as large as possible.
Such an algorithm leverages prior knowledge when the predicted prior is correct, while retaining meaningful worst-case guarantees.

The $\nicefrac{1}{\ee}$-competitive Dynkin's algorithm is $\nicefrac{1}{\ee}$-consistent and $\nicefrac{1}{\ee}$-robust for both objectives.
The optimal $\alpha^*_{\MaxExp} \approx 0.745$-competitive \MaxExp algorithm by~Hill and Kertz~\cite{hillComparisonsStopRule1982,kertzStopRuleSupremum1986} and $\alpha^*_{\MaxProb} \approx 0.580$-competitive \MaxProb algorithm by~\citet{gilbertRecognizingMaximumSequence}  are $\alpha^*_{\MaxExp}$ and $\alpha^*_{\MaxProb}$-consistent, respectively, but only $0$-robust (see \Cref{sec:prelim}).
Further, randomly running either Dynkin's algorithm or an i.i.d.\ prophet inequality algorithm linearly interpolates between their consistency and robustness.
Finally, the optimality of these competitive ratios implies that no algorithm is better than $\nicefrac{1}{\ee}$-robust, or better than $\alpha^*_\MaxExp$-consistent for \MaxExp, or better than $\alpha^*_\MaxProb$-consistent for \MaxProb.
These baseline algorithms and impossibilities leave a significant gap between them.

\subsection{Overview of Results}

This paper introduces  two new algorithms for the \MaxExp and \MaxProb objectives respectively, under a unified framework, and a hardness result for \MaxProb.
See \Cref{fig:Tradeoff} for an illustration.

\begin{figure}[h]
\centering
\begin{subfigure}{0.45\textwidth}    
    \includegraphics[width=\textwidth]
    {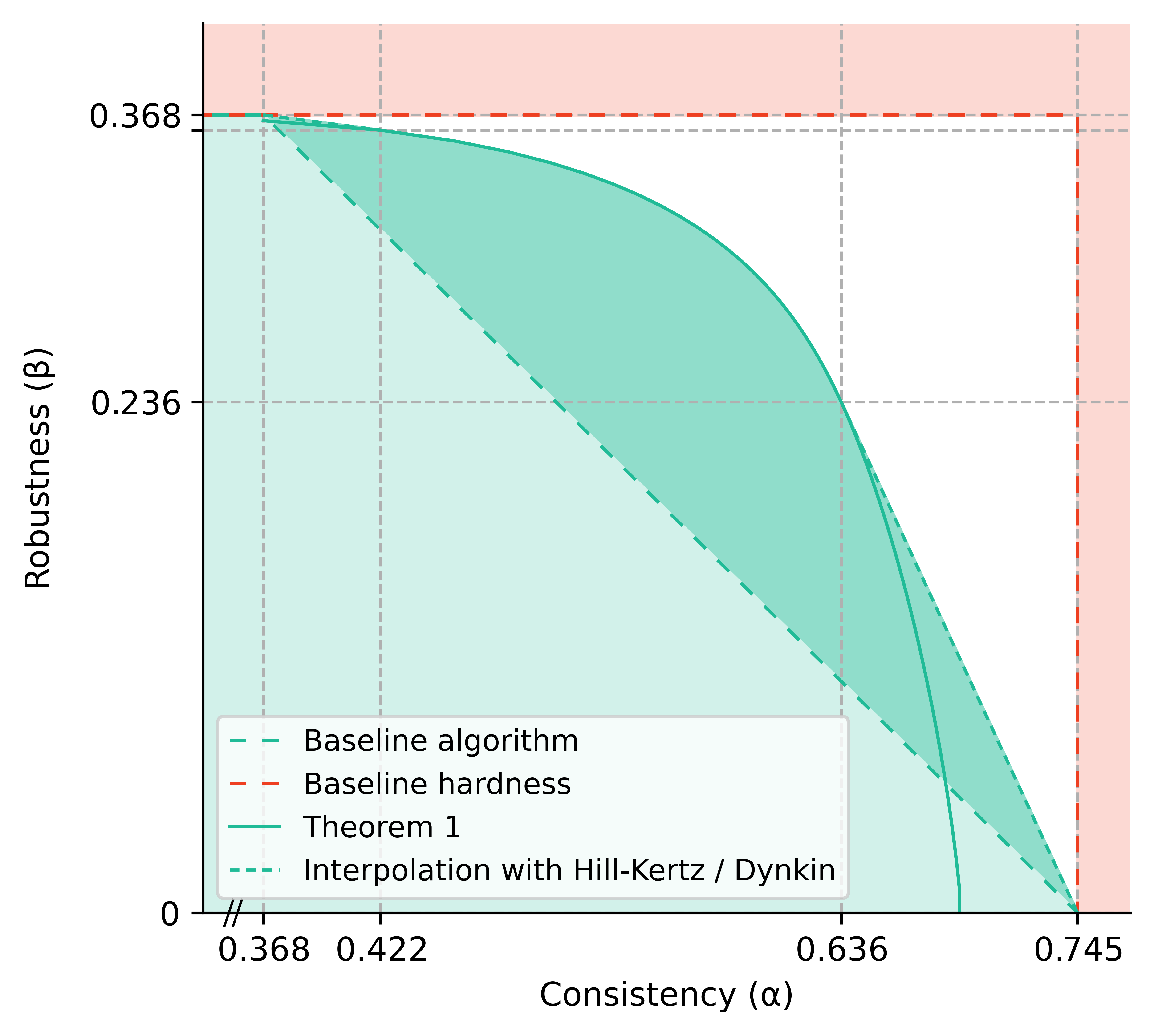}
    \caption{\MaxExp}
    \label{fig:MaxExp}
\end{subfigure}
\begin{subfigure}{0.45\textwidth}
    \includegraphics[width=\textwidth]%
    {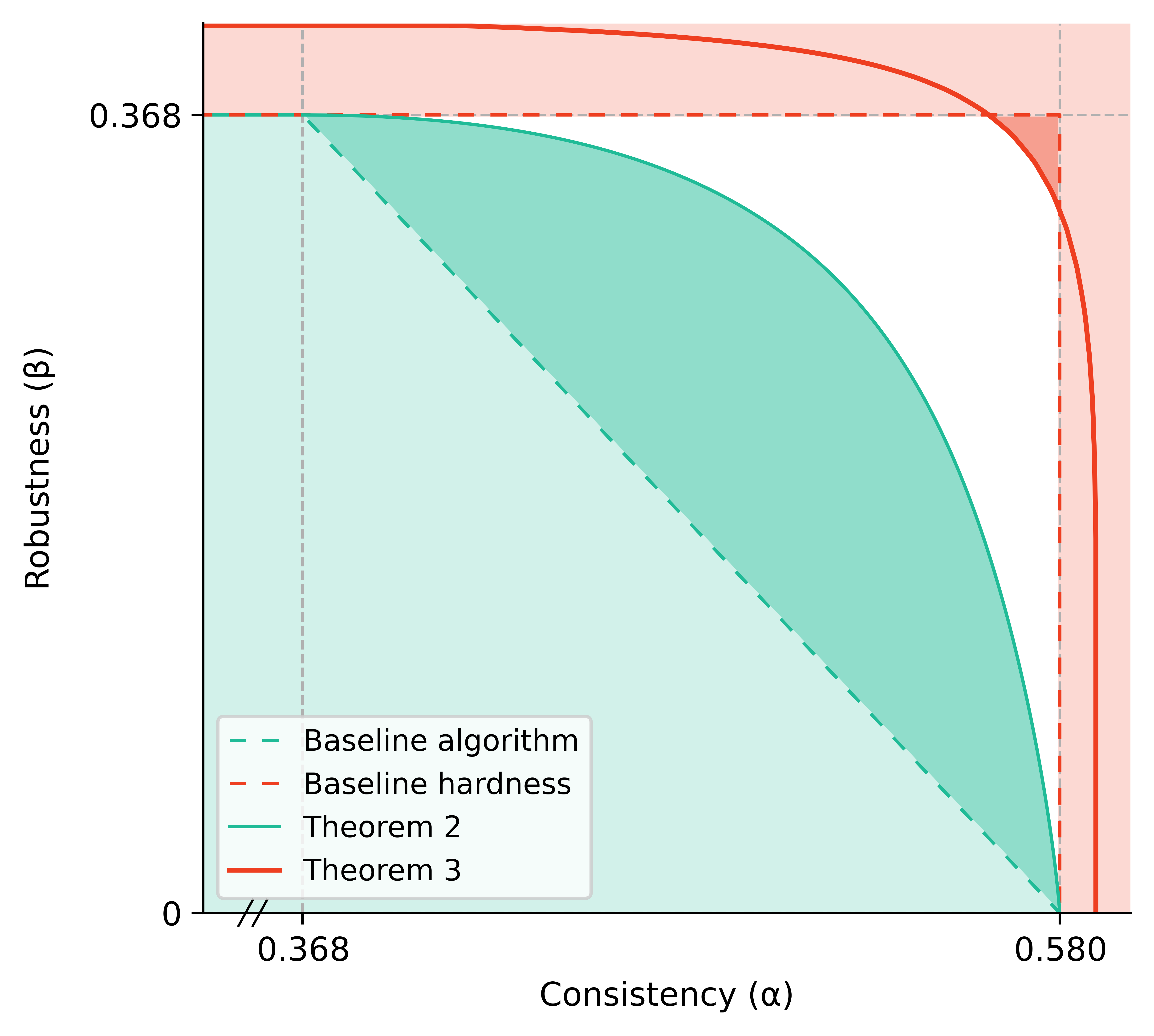}
    \caption{\MaxProb}
    \label{fig:MaxProb}
\end{subfigure}
        
\caption{Illustration of our results and a comparison with the baselines}
\label{fig:Tradeoff}
\end{figure}

\vspace{-10pt}

\paragraph{Bi-criteria Algorithms.}
We define a family of bi-criteria algorithms $\algowiththres(\prior, \threshold)$ parameterized by the predicted prior $\tilde{F}$ and a threshold function $\threshold$.
Upon observing each value and when no value has been accepted, a bi-criteria algorithm accepts the current value if:
\begin{enumerate}
    \item It is a record, i.e., greater than or equal to all previous values; \emph{(prior-free criterion)}
    \item Its cdf value w.r.t.\ $\tilde{F}$ exceeds a time-dependent threshold $\threshold(t)$. \emph{(prior-dependent criterion)} 
\end{enumerate}
For now, readers may interpret $t$ as the time step by which the value is observed.
Later, we will also consider an equivalent continuous-time model to simplify the analysis (see~\Cref{sec:prelim} for a discussion).

We set the thresholds as follows.
Before some time $\lambda_1$, the threshold is $1$, and the algorithm rejects all values.
After some time $\lambda_2$, the threshold is $0$, and the algorithm's decisions depend solely on the prior-free record-based criterion.
These two prior-free phases yield robustness.
In fact, we derive this design from first principles, as we will explain shortly in \Cref{sec:techniques}.

We obtain new algorithmic results for both \MaxExp and \MaxProb under this framework.

\begin{theorem}[\MaxExp algorithm]
\label{thm:max-exp-consistency-robustness}
    For any $\beta \in [0, \nicefrac{1}{e}]$, suppose $\lambda_1, \lambda_2$ are the two roots of $-\lambda \ln \lambda = \beta$ with $\lambda_1 \le \lambda_2$.
    If there is a left-continuous non-increasing function $\threshold \colon [\lambda_1, \lambda_2] \to [0, 1]$ satisfying:
    \begin{equation}
    \label{eqn:thres-max-expect}
    \begin{aligned}
        \threshold(\lambda_1) & ~=~ 1 ~, \\
        \int_z^1 \int_0^t \, \frac{1}{t} \threshold ( \max \{s, z \} )^t \dd s \dd t 
        & ~\ge~ \alpha \cdot \threshold(z)
        ~, \qquad 
        \forall z \in[\lambda_1, \lambda_2] 
        ~,
    \end{aligned}
    \end{equation}
    then for any $n\in \mathbb{N}$, there is a bi-criteria algorithm that is $\alpha$-consistent and $\beta$-robust for the \MaxExp objective.
\end{theorem}

The proof of~\Cref{thm:max-exp-consistency-robustness} is given in~\Cref{sec:algo-max-expect}.
To plot the curve in \Cref{fig:MaxExp}, we rely on numerical methods to solve the differential equation~\eqref{eqn:thres-max-expect}, which does not appear to admit a closed-form solution.
See~\Cref{app:max-expect-numerical-solution}.

\begin{theorem}[\MaxProb algorithm]
    \label{thm:max-prob-consistency-robustness}
    For any $\beta \in [0, \frac{1}{e}]$, suppose $\lambda_1, \lambda_2$ are the two roots of $-\lambda \ln \lambda = \beta$ with $\lambda_1 \le \lambda_2$.
    For any $n\in \mathbb{N}$, there is an $\alpha$-consistent $\beta$-robust bi-criteria algorithm for the \MaxProb objective, with
    \begin{equation*}
        \alpha = \beta  + \int_{\lambda_1}^{\lambda_2} \int_s^1 \,
        \frac{\ee^{-\frac{\constint t}{1-s}}}{t} \dd t \dd s
    ~,
    \end{equation*}
    where $\constint \approx 0.804$ is the solution to $\sum_{k=1}^{\infty} \frac{\constint^k}{k!k} = 1$.
    Moreover, this trade-off is tight within the class of bi-criteria algorithms in the sense that, for any specified $(\alpha,\beta)$, none of the bi-criteria algorithms can improve both objectives simultaneously.
\end{theorem}

The proof of~\Cref{thm:max-prob-consistency-robustness} is given in~\Cref{sec:algo-MaxProb}.

For the \MaxProb objective, besides the tightness within the bi-criteria family above, we further show a non-trivial hardness that certain consistency-robustness trade-offs cannot be achieved by \emph{any} algorithm.
In particular, this implies that there is no ``best-of-both-world'' algorithm for this objective.

\begin{theorem}[\MaxProb hardness, informal version of \Cref{thm:hardness-maxprob}]
There is an $\alpha$-consistent $\beta$-robust algorithm for the \MaxProb objective only if a polytope $\mathcal{P}(\alpha, \beta)$ is feasible.
In particular, no algorithm is both $\alpha^*_\MaxProb$-consistent and $\nicefrac{1}{\ee}$-robust.
\end{theorem}

Again, to plot the curve in~\Cref{fig:MaxProb}, we rely on numerical linear program solvers to solve the feasibility of certain polytopes, which does not appear to admit a closed-form solution.

\subsection{Overview of Techniques}
\label{sec:techniques}

\paragraph{From Mispredicted Prior to Misspecified Thresholds.}
A key property of bi-criteria algorithms is that they provide two equivalent viewpoints for robustness. 
By definition, we may consider a real prior $F$ and a (mis)predicted prior $\tilde{F}$ that misleads the algorithm.
Equivalently, we may also consider an algorithm with the correct prior $F$ but a misspecified threshold function $F \circ \tilde{F}^{-1} \circ \threshold$, since $\tilde{F}(x) \ge \threshold(t)$ if and only if $F(x) \ge F \circ \tilde{F}^{-1} \circ \threshold(t)$.
The latter viewpoint is much easier to work with --- we just need to restrict the misspecified threshold $F \circ \tilde{F}^{-1} \circ \threshold(t)$ via our choice of threshold $\threshold$.
Further, the only restrictions are that whenever $\threshold(t)$ is $0$ or $1$, $F \circ \tilde{F}^{-1} \circ \threshold(t)$ should have the same value of $0$ or $1$.
This observation leads to our design of thresholds with three phases:
1) a prior-free phase before time $\lambda_1$ where $\threshold(t) = 1$ and the algorithm rejects all values;
2) a prior-dependent phase from time $\lambda_1$ to $\lambda_2$ where $\threshold(t)$ lies between $0$ and $1$; and
3) a final prior-free phase after time $\lambda_2$ where $\threshold(t) = 0$ and the algorithm takes any record.
See \Cref{sec:algo} for a detailed discussion, and \Cref{sec:robustness} for the robustness analysis.

\paragraph{Implicit Sharding.}
We show that optimal stopping with a predicted prior is the hardest when the number of values tends to infinity, using an \emph{implicit sharding} technique that may be of independent interest.
This is inspired by the \emph{sharding} technique
by 
\citet{harbNewProphetInequalities2025}, which breaks each value into many i.i.d.\ shards based on the prior, such that the original value equals the maximum shard. 
However, the sharding technique requires full knowledge of the real prior, and thus, does not apply to optimal stopping with a predicted prior.
In contrast, implicit sharding directly works with the original values, and only uses the shards implicitly in the analysis through a coupling argument. 
See \Cref{sec:sharding} and specifically \Cref{thm:max-exp-sharding} for the formal argument.

\paragraph{Tractable Factor-revealing Linear Programs.}
We prove the hardness result by constructing factor-revealing linear programs whose optimal values upper-bound the achievable consistency-robustness trade-offs.
Exponential-size linear programs can easily capture all possible algorithms, which map possible sequences of observed values to binary acceptance decisions. 
The main challenge is to reduce the size without losing much characterization power.

We compress the algorithms' state space by restricting the real prior to be a conditional distribution of the predicted prior, subject to an upper bound $k$ on the values.
Intuitively, the only information an algorithm may conclude from the observed values $x_1, \ldots, x_t$ is $k \ge \max_{1 \le i \le t} x_i$ under this setup.
We formalize this intuition, showing that we may, without loss of generality, consider \emph{minor-oblivious} algorithms whose decisions depend solely on the maximum observed value, not on the detailed observed value sequence.
This reduces the state space from an exponential number of possible value sequences to a polynomial number of possible maximum values.
See \Cref{sec:hardness-MaxProb} for the factor-revealing linear programs and the hardness.

\subsection{Related Work}

\paragraph{Prophet Inequalities with Full Prior Information.}
When the algorithm has complete knowledge of the distributions, the prophet inequality problem has been extensively studied.
Under the \MaxExp objective, \citet{hillComparisonsStopRule1982} proposed a $(1 - \nicefrac{1}{\ee})$-competitive algorithm for i.i.d.\ prophet inequality.
This was improved to $0.738$ by \citet{abolhassaniBeating11Ordered2017}, and finally to approximately $0.745$ by \citet{correaPostedPriceMechanisms2017}. 
The latter ratio is tight, matching a hardness result by Hill and Kertz~\cite{hillComparisonsStopRule1982,kertzStopRuleSupremum1986}.
For prophet inequality where values are independently drawn from non-identical distributions, the optimal competitive ratio is $0.5$~\cite{krengel1978semiamarts, samuel1984comparison}.
The \emph{prophet secretary problem} \cite{esfandiariProphetSecretary2015} interpolates between these two settings: values are independently drawn from non-identical distributions and arrive in random order.
The optimal competitive ratio for prophet secretary is open: a lower bound of $0.688$ \cite{chenProphetSecretaryMatching2025} and an upper bound of $0.723$~\cite{giambartolomeiProphetInequalitiesSeparating2023} are currently known.

For the \MaxProb objective, \citet{gilbertRecognizingMaximumSequence} gave an optimal algorithm for the i.i.d.\ prophet inequality, while \citet{samuels1982} derived an explicit expression for the optimal competitive ratio of approximately $0.580$. 
For non-identical distributions, \citet{esfandiariProphetsSecretariesMaximizing2020} obtained an optimal $\nicefrac{1}{\ee}$-competitive ratio for the prophet inequality model, and recovered the $0.580$ ratio for the prophet secretary under a ``no-superstar'' assumption. 
\citet{nutiSecretaryProblemDistributions2022} later removed this assumption.

See \citet{correaRecentDevelopmentsProphet2019} for a survey.
We stress that none of these prophet inequality algorithms provide non-trivial robustness guarantees.

\paragraph{Prophet Inequalities with Partial Prior Information.}
Relaxing the assumption of complete prior knowledge, several models have studied algorithms with partial information about the distributions.
The most prominent one is the \emph{sample-based model}, in which the algorithm has access to i.i.d.\ samples from the prior distribution(s). 
One line of research focuses on getting constant-factor prophet inequalities using one (or fewer) sample from each distribution~\cite{azar2014prophet, rubinsteinOptimalSinglechoiceProphet2020, gravinOptimalProphetInequality2022, ezraProphetInequalitySamples2024}.
Notably, the optimal $\nicefrac{1}{2}$-competitive prophet inequality can be achieved with only one sample from each distribution~\cite{rubinsteinOptimalSinglechoiceProphet2020, gravinOptimalProphetInequality2022}.
Another line of work studies the minimal number of samples needed to achieve the optimal competitive ratio up to $\varepsilon$~\cite{diakonikolasLearningOnlineAlgorithms2021, correaUnknownIIDProphets2021, correaTwosidedGameGoogol2020, rubinsteinOptimalSinglechoiceProphet2020, cristiProphetInequalitiesRequire2023, guo2021generalizing, correaSampledrivenOptimalStopping2021, jin2024sample}.
The state-of-the-art bounds for i.i.d.\ prophet inequality are $\tilde{\bigO}(\nicefrac{n}{\varepsilon})$ for learning a $0.745-\varepsilon$ competitive algorithm~\cite{correaSampledrivenOptimalStopping2021}, and $\tilde{\bigO}(\nicefrac{n}{\varepsilon^2})$ for learning an online policy that is optimal up to an $\varepsilon$ additive error~\cite{guo2021generalizing}. 
For general prophet inequality, $\tilde{\bigO}(\nicefrac{1}{\varepsilon^2})$ samples per distribution suffice~\cite{jin2024sample}.
Finally, \citet{cristiProphetInequalitiesRequire2023} gave a unified proof that $\bigO(\mathrm{poly}(\nicefrac{1}{\varepsilon}))$ samples per distribution suffice for two variants of prophet inequalities.
Other models of partial information include quantile-query oracles~\cite{liProphetInequalityIID2023,perezsalazarIIDProphetInequality2023} and bandit feedback~\cite{gatmiry2024bandit}.
These models and the sample-based model are incomparable to ours:
they assume access to the \emph{correct prior}, while we consider a \emph{predicted prior} that may be \emph{arbitrarily wrong}.

\paragraph{Algorithms with Predictions.}

Our work falls within the \emph{algorithms with predictions} framework~\cite{mitzenmacher2022algorithms} for augmenting worst-case analysis of algorithms with machine-learned advice. 
Consistency, robustness, and smoothness are the primary performance metrics, capturing the algorithms' performance when the prediction is correct, completely wrong, and partially wrong.
This framework has been applied to many online optimization problems, including ski-rental~\cite{purohit2018improving,diakonikolasLearningOnlineAlgorithms2021,benomar2023advice,angelopoulos2024online}, scheduling~\cite{purohit2018improving,benomar2023advice}, caching~\cite{lykouris2021competitive}, online bidding~\cite{angelopoulos2024online}, and so on.

Within optimal stopping, prior work has explored several forms of predictions. \citet{antoniadisSecretaryOnlineMatching2020} studies a prediction of the maximum value, \citet{fujiiSecretaryProblemPredictions2024} and \citet{balkanski2024fair} assume having a predicted sequence of values, and \citet{braun2024secretary} considers a prediction of the gap between the maximum and the $k$-th largest value.

In contrast, our work considers \emph{distributional predictions}, bridging the extensive literature of prophet inequality with the growing use of machine-learned priors in practice. Within the context of prophet inequalities, the only relevant work to our knowledge is by~\citet{duttingPostedPricingProphet2019}, which examines the impact of mispredicted priors with an emphasis on smoothness. This work is the first to investigate the consistency-robustness trade-off given a distributional prediction.
In general, while distributional predictions have been explored in certain domains such as online matching, where they have been incorporated into traditional stochastic models \citep{esfandiari2015online,canonne2025little}, this type of prediction remains largely underexplored.

\paragraph{Independent and Concurrent Work.}
\citet{kehne2025prophet} study the same problem, focusing on the MaxExp setting. 
Their results and ours are complementary.
The algorithms in both papers are under the bi-criteria framework.
They show a non-trivial consistency-robustness trade-off using a simple single-threshold between $\lambda_1$ and $\lambda_2$, whereas we optimize the choice of thresholds within the bi-criteria framework.
They rule out ``best-of-both-world'' algorithms for the \MaxExp setting, whereas we prove such an impossibility result for the \MaxProb setting.

\section{Preliminaries}
\label{sec:prelim}

Let $\N$ denote the set of positive integers.
For any $n \in \N$, define $[n] \coloneqq \{1, 2, \ldots, n\}$.
Let $\mathcal{D}$ be the set of all probability distributions supported on $[0, \infty)$. 
We adopt the standard abbreviations \emph{cdf}, \emph{pdf}, and \emph{pmf} for the cumulative distribution function, probability density function, and probability mass function. 
We use the same symbol for a distribution and its cdf (e.g., $F$).
Unless stated otherwise, expectations and probabilities are taken over the randomness of the values and any internal randomness of the algorithm.

\subsection{Optimal Stopping}

Consider $n$ values $x_1, x_2, \ldots, x_n$ drawn i.i.d.\ from a prior distribution $F \in \mathcal{D}$.
The values arrive one by one.
Upon the arrival of each value $x_i$, the algorithm observes the realized value, and then must immediately decide whether to accept $x_i$ and stop, or to reject it and continue.
At most one value may be accepted.

Fix any algorithm, let $\AlgVal$ be the index of the value it accepts, and thus, $\ALG$ be the corresponding accepted value.
We define $\ALG = 0$ if the algorithm rejects all values.

We consider the maximization problem, where larger values are preferable.
The literature has studied two objectives.
Let $\OPT$ denote the maximum value among all $n$ values, i.e., $\OPT \coloneqq \max_{i\in [n]} x_i$.
The \emph{max expectation} (\MaxExp) objective considers the expectation of the algorithm's accepted value, $\E [ \ALG ]$.
The corresponding \emph{competitive ratio} is
\[
    \frac{\E[\ALG]}{\E[\OPT]}
    ~.
\]
The \emph{max probability} (\MaxProb) objective considers the probability of accepting the maximum:
\[
    \Pr [ \ALG = \OPT ]
    ~.
\]
This coincides with the \emph{competitive ratio} since the offline optimal choice is always the maximum.
An algorithm is $\Gamma$-competitive if its competitive ratio is at least $\Gamma$.

\subsection{Optimal Stopping with a Predicted Prior}

We assume that the number of values $n$ is known.
Depending on what the algorithm knows about the prior distribution $F$, we have different optimal stopping problems.
If nothing is known, we have essentially the secretary problem~\cite{dynkinOptimumChoiceInstant1963, fergusonWhoSolvedSecretary1989}.%
\footnote{The classical secretary problem considers random order without prior.
Nonetheless, \citet{correaProphetInequalitiesIID2019} showed that the optimal competitive ratio is still $\nicefrac{1}{\ee}$ for an unknown prior.}
With complete information of $F$, we have the i.i.d.\ prophet inequality~\cite{correaUnknownIIDProphets2021}.

This paper considers a new model between the two extremes, which we dub \emph{optimal stopping with a predicted prior}.
In this model, the algorithm is given a possibly erroneous predicted prior $\prior$.
Hence, upon receiving each value $x_i$, the algorithm decides whether to accept based on the number of values $n$, the predicted prior $\prior$, and the realizations of the values observed thus far including $x_i$.

Following the literature of algorithms with predictions \cite{mitzenmacher2022algorithms}, we define the consistency and robustness of stopping algorithms.
\begin{definition}[Consistency]
An algorithm is \emph{$\alpha$-consistent} if for any prior distribution $F \in \mathcal{D}$, it is $\alpha$-competitive when the predicted prior is accurate, i.e., when $\prior = F$.
\end{definition}

\begin{definition}[Robustness]
An algorithm is \emph{$\beta$-robust} if for any prior distribution $F \in \mathcal{D}$, it is $\beta$-competitive regardless of which predicted prior $\prior$ is given.
\end{definition}

\subsection{Simplifying Assumptions for Algorithm Analyses}
\label{subsec:assumption-time-and-tiebreak}

For our algorithmic results, we adopt three standard assumptions from the literature that simplify the analysis and are without loss of generality.

\paragraph{Continuous-time  Model.}
Besides the above classical model, which we will refer to as the discrete-time model, we also consider an equivalent continuous-time model. In this variant, each value $x_i$ is further assigned an arrival time $t_i$ drawn independently and uniformly from $[0, 1]$, at which time the algorithm observes the value.
The induced order of arrivals is a uniformly random permutation of the $n$ values, so the discrete-time and continuous-time formulations are equivalent.

\paragraph{Random Tie-breaking.}
We assume the predicted prior $\tilde{F}$ is continuous. This assumption is without loss of generality for our algorithm, since we can apply the random tie-breaking rule by \citet{correaProphetSecretaryBlind2021}.

\paragraph{Full-support Continuous Prior.}
For the analysis, we further assume that the prior $F$ is continuous with full support on $[0, \infty)$.
This assumption is without loss of generality.
To see this, we construct a perturbed distribution
\[
    F_{\varepsilon}=
    \begin{cases}
        F &\text{ w.p. } 1-\varepsilon~; \\
        \text{Exp}(1) &\text{ w.p. } \varepsilon~.
    \end{cases} 
\]
Then $F_{\varepsilon}$ is continuous with full support.
Moreover, with probability at most $1 - (1 - \varepsilon)^n \leq n \varepsilon$, the instance observed by the algorithm will be changed.
Thus, for any finite $n\in \mathbb{N}$, taking $n \varepsilon \to 0$, the performance is preserved.

For any distribution $F \in \mathcal{D}$, define its right‑continuous inverse as:
\[
    F^{-1}(q) \coloneqq \inf \big\{ x \geq 0 : F(x) \ge q \big\}, \quad \forall q \in [0, 1].
\]
Importantly, $F$ and $F^{-1}$ are strictly increasing on the interiors of their domains.
Further, we have the boundary conditions $F^{-1}(0) = 0$ and $F(0) = 0$, and by convention $F^{-1}(1) = \infty$ and $F(\infty) = 1$.

\section{Bi-criteria Algorithms}
\label{sec:algo}

We start with two definitions that will be repeatedly used throughout the paper.

\begin{definition}
The \emph{prefix maximum} for any time $t \in [0, 1]$ is the maximum value arrived before time $t$, i.e., $\max_{i \,:\, t_i < t} x_i$, with the convention that it equals $0$ if no value arrives before $t$.
\end{definition}

\begin{definition}
    A value $x$ arriving at time $t$ is called a \emph{record} if $x$ is greater than or equal to the prefix maximum for time $t$.
\end{definition}

We consider the following family of algorithms parameterized by the predicted prior $\prior$ and a threshold function $\threshold$. 
Throughout the paper, threshold functions are always \emph{non-increasing} maps from $[0, 1]$ to $[0, 1]$.

\medskip

\begin{tcolorbox}[title={Bi-criteria Algorithm $\algowiththres(\prior, \threshold)$}]
    \label{algo:algo-framework}
    For each value $x_i$ arriving at time $t_i$ when no value has been accepted, accept $x_i$ if:
    \begin{enumerate}
        \item Value $x_i$ is a record; and 
        \item The cdf value of $x_i$ w.r.t.\ $\prior$ is greater than the threshold at time $t_i$, i.e., if $\prior(x_i) > \threshold(t_i)$.
    \end{enumerate}
\end{tcolorbox}

\medskip

The first prior-free criterion is without loss of optimality for \MaxProb and is critical for the robustness guarantees for both \MaxProb and \MaxExp.
The second prior-dependent criterion represents how the algorithm uses the predicted prior to guide its decisions.

Specifically, if $\threshold(t)=1$, the algorithm does not accept any values at time $t$, regardless of how large the values are.
If $\threshold(t)=0$, the algorithm accepts any record at time $t$, regardless of the value's relative standing w.r.t.\ the predicted prior. 

Algorithms under this framework satisfy many properties, including a simple necessary and sufficient condition for the algorithm to reject all values before time $t$, as follows. As discussed in~\Cref{subsec:assumption-time-and-tiebreak}, we assume the distribution is continuous and thus ties occur with probability zero.

\begin{proposition}
    \label{prop:no-acceptance}
    The algorithm $\algowiththres(\tilde{F}, \threshold)$ rejects all values before time $t$ if and only if the prefix maximum $y$ for time $t$ and the arrival time $s$ of $y$ satisfy $\prior(y) \leq \threshold(s)$.
\end{proposition}

\begin{proof}
    If $\prior(y) \leq \threshold(s)$, then all values $x$ arriving before time $s$ are rejected because the cdf value is below the threshold, as $\prior(x) \le \prior(y) \le \threshold(s)\le \threshold(t)$ for any $t\leq s$.
    Values arriving between $s$ and $t$ are also rejected because they do not exceed $y$ and therefore cannot be a record.

    If $\prior(y) > \threshold(s)$, then $y$ is a record and has cdf value above the threshold. 
    Hence, the algorithm accepts either $y$ or a value before $y$.
\end{proof}

\begin{corollary}
    \label{prop:acceptance-rule}
    Suppose a value $x$ arrives at time $t$, and its prefix maximum $y$ arrives at time $s < t$.
    The algorithm accepts $x$ if and only if 
    1) $\prior(x) > \threshold(t)$, 
    2) $x \geq y$, and 
    3) $\prior(y)\leq \threshold(s)$.
\end{corollary}

\subsection{Examples}

This family encompasses many classical algorithms as special cases.
We present three examples.

\begin{example}[Dynkin's algorithm~\cite{dynkinOptimumChoiceInstant1963}, continuous-time variant]
\label{example:dynkins}
    Dynkin's algorithm accepts the first record after rejecting the first $\lfloor \nicefrac{n}{\ee}\rfloor$ values. 
    The continuous-time variant (e.g., \citet{feldmanImprovedCompetitiveRatios2011}) accepts the first record after time $\nicefrac{1}{\ee}$, which corresponds to the threshold function $\threshold(t) = 1$ for $t \in [0, \nicefrac{1}{\ee}]$, and $\threshold(t) = 0$ for $t \in (\nicefrac{1}{\ee}, 1]$.
\end{example}

Dynkin's algorithm is $\nicefrac{1}{\ee}$-competitive for both \MaxProb and \MaxExp, without using any information about the prior~\cite{dynkinOptimumChoiceInstant1963, feldmanImprovedCompetitiveRatios2011}.
Hence, it is $\nicefrac{1}{\ee}$-consistent and $\nicefrac{1}{\ee}$-robust.

\begin{example}[Gilbert-Mosteller algorithm~\cite{gilbertRecognizingMaximumSequence}, continuous-time variant]
\label{example:gilbert-mosteller}

With a correct prior $\prior = F$, the optimal algorithm for the \MaxProb objective takes the bi-criteria form according to backward induction.
In this setting, Gilbert and Mosteller calculated the optimal thresholds in the discrete-time model.
Their algorithm achieves a competitive ratio of approximately $0.580$, with the following analytical form due to \citet{samuels1982}:\footnote{
The competitive ratio presented in \citet{samuels1982} is $\ee^{-\constint}+(\ee^\constint - \constint - 1)\int_1^\infty \frac{\ee^{-\constint x}}{x}\dd x$ (fixing a typographical negative sign before the second term), which equals our expression through transformations.
}
\[
    \int_{0}^{1} \int_s^1
    \frac{\ee^{-\frac{\constint t} {1-s}}}{t} \dd t \dd s 
    ~,
\]
where $\constint \approx 0.804$ is the solution to $\sum_{k=1}^{\infty} \frac{\constint^k}{k!k} = 1$.

Using a similar backward induction argument, we compute in~\Cref{lem:unique-inheritance-optimal-threshold} the optimal threshold $\maxproboptthres$ under our bi-criteria framework for any $n\in \mathbb{N}$. For each $t$, $\maxproboptthres(t)$ is the unique solution $\quantile$ of:
\begin{equation*}
    \int_t^1 \frac{(1 - s + s\quantile)^{n-1} - \quantile^{n - 1}}{1 - s} \dd s = \quantile^{n-1}
    ~.
\end{equation*}
\end{example}

This construction achieves the same consistency $0.580$ for \MaxProb as Gilbert-Mosteller algorithm, and it is exactly the form used in the proof of \Cref{thm:max-prob-consistency-robustness}.
However, the algorithm is merely $0$-robust.
To illustrate, consider the case where the prior $F$ is uniform over $[0,1]$, whereas the predicted prior $\prior$ is uniform over $[2,3]$.
In this case, the algorithm rejects all values, as they have cdf value $0$ w.r.t.\ the predicted prior.

\begin{example}[Single-threshold algorithm~\cite{ehsaniProphetSecretaryCombinatorial2018}, i.i.d.\ case]
    Let $\threshold(t) = 1 - \tfrac{1}{n}$ for $t \in [0, 1]$.%
    \footnote{The original algorithm sets a constant threshold $\threshold$ such that $\Pr[\OPT \geq \threshold] = \nicefrac{1}{\ee}$, which corresponds to $F(x_i) \approx 1 - \nicefrac{1}{n}$ for each $i$ in the i.i.d.\ case.}
\end{example}

The single-threshold algorithm achieves a competitive ratio of $1-\nicefrac{1}{\ee} \approx 0.632$ for the prophet secretary problem in terms of \MaxExp.
Hence, it is $(1-\nicefrac{1}{\ee})$-consistent for \MaxExp.
However, it suffers from the same counterexample and is only $0$-robust like Gilbert-Mosteller algorithm.

\medskip

In sum, the threshold functions corresponding to existing algorithms fail to achieve nontrivial trade-offs between consistency and robustness. 
They either ignore the predicted prior and have pessimistic consistency (Dynkin's algorithm) or rely on having an accurate prior without any robustness (Gilbert-Mosteller algorithm and single-threshold algorithm).

\subsection{Robust Threshold Functions}
\label{subsec:robust-threshold-function}

We next explain our design of robust threshold functions.
First, we need to understand how an adversary would choose a worst-case prior distribution $F$, given the predicted prior $\prior$, threshold function $\threshold$, and the corresponding algorithm $\algowiththres\big(\prior, \threshold\big)$.
The key observation is the following lemma.

\begin{lemma}[From mispredicted prior to misspecified thresholds]
    \label{obs:transformation}
    For any distributions $F$ and $\prior$, and any threshold function $\threshold$, algorithms $\algowiththres\big(\prior ,\threshold\big)$ and $\algowiththres(F, F \circ \prior^{-1} \circ \threshold)$ are the same.
\end{lemma}

\begin{proof}
The lemma holds because $\prior(x) \,{>}\, \threshold(t)$ if and only if $F(x) \,{>}\, F \circ \prior^{-1} \circ \threshold(t)$, according to the monotonicity of $F$ and $\prior^{-1}$ and their boundary conditions.
\end{proof}

Despite its simple one-sentence proof, the lemma has an important implication.
Instead of considering algorithm $\algowiththres\big(\prior, \threshold\big)$ parameterized with a \emph{mispredicted prior} $\prior$, we may equivalently consider algorithm $\algowiththres\big(F, \threshold\big)$ parameterized with the correct prior $F$ and a \emph{misspecified threshold function} $F \circ \prior^{-1} \circ \threshold$.

Importantly, the misspecified threshold function $F \circ \prior^{-1} \circ \threshold$ is not arbitrary.
By the boundary conditions of $F$ and $\prior^{-1}$, if the original threshold $\threshold(t)$ equals $0$ or $1$ at time $t$, the misspecified threshold function $F \circ \prior^{-1} \circ \threshold (t)$ must have the same value $0$ or $1$.
Further, it is easy to show that this is the only constraint.
Let $\lambda_1 = \sup \{ 0 \le t \le 1 : \threshold(t) = 1 \}$ and $\lambda_2 = \inf \{ 0 \le t \le 1 : \threshold(t) = 0 \}$.
Then, the misspecified threshold function $F \circ \prior^{-1} \circ \threshold(t)$ can be any non-increasing function for $t$ between $\lambda_1$ and $\lambda_2$, while having value $1$ for $0 \le t < \lambda_1$ and value $0$ for $ \lambda_2 < t \le 1$.

In Section~\ref{sec:robustness}, we will further show that the \emph{worst-case} threshold $F \circ \prior^{-1} \circ \threshold(t)$ is either always $0$ or always $1$ in $[\lambda_1, \lambda_2]$.
They correspond to a variant of Dynkin's algorithm that starts to accept records at a suboptimal time $\lambda = \lambda_1$ or $\lambda_2$, rather than the optimal choice $\nicefrac{1}{\ee}$, whose competitive ratio is $-\lambda \ln \lambda$ for both \MaxProb and \MaxExp (e.g., by \citet*{feldmanImprovedCompetitiveRatios2011}). 

This yields the following robustness guarantee, whose proof is deferred to Section~\ref{sec:robustness}.

\begin{lemma}
\label{lem:pre-robustness}
For any predicted prior $\prior$, any threshold function $\threshold$, denote $\lambda_1 = \sup \{ 0 \le t \le 1 : \threshold(t) = 1 \}$ and $\lambda_2 = \inf \{ 0 \le t \le 1 : \threshold(t) = 0 \}$, algorithm $\algowiththres\big(\prior,\threshold\big)$ is
\[
\min\{-\lambda_1 \ln \lambda_1, -\lambda_2 \ln \lambda_2 \}\text{-robust}
\]
for both \MaxProb and \MaxExp.
\end{lemma}

Motivated by \Cref{lem:pre-robustness}, we now define the robust counterpart
of any threshold function.

\begin{definition}
\label{def:robust-threshold}

For any $0\leq \lambda_1\leq \lambda_2\leq 1$, for any threshold function $\threshold$, define its corresponding \emph{robust threshold function} as:
\[
    \restrict{\threshold}(t) \coloneqq 
    \begin{cases}
    1 & 0 \le t \le \lambda_1 ~; \\
    \threshold(t) & \lambda_1 < t \le \lambda_2 ~; \\
    0 & \lambda_2 < t \le 1 ~.
    \end{cases}
\]
\end{definition}

The robustness guarantee in Lemma~\ref{lem:pre-robustness} depends only on the
pair $(\lambda_1,\lambda_2)$, which allows us to parameterize robustness
directly. 
Consider any feasible robustness parameter $\beta \in [0, \nicefrac{1}{e}]$.
Let $0 \le \lambda_1 \le \nicefrac{1}{\ee} \le \lambda_2 \le 1$ be the solutions to:
\[
    -\lambda \ln \lambda ~=~ \beta
    ~.
\]
\Cref{lem:pre-robustness} above immediately implies that any algorithm with the form $\algowiththres(\prior,\restrictsmall{\threshold})$ is $\beta$-robust for both \MaxProb and \MaxExp.

Readers familiar with the Lambert $\mathcal{W}$ function may note that:
\begin{equation}
    \label{eqn:lambert-W}
    \lambda_1 \coloneqq \Wone~,
    \quad\qquad \lambda_2 \coloneqq \Wzero
    ~.
\end{equation}

As an illustration, \Cref{fig:max-prob-threshold} shows the robust counterpart of the threshold function from the Gilbert-Mosteller algorithm with robustness $\beta=\nicefrac{1}{3}$.

\begin{figure}[htbp!]
    \centering
    \includegraphics[width=0.4\textwidth]{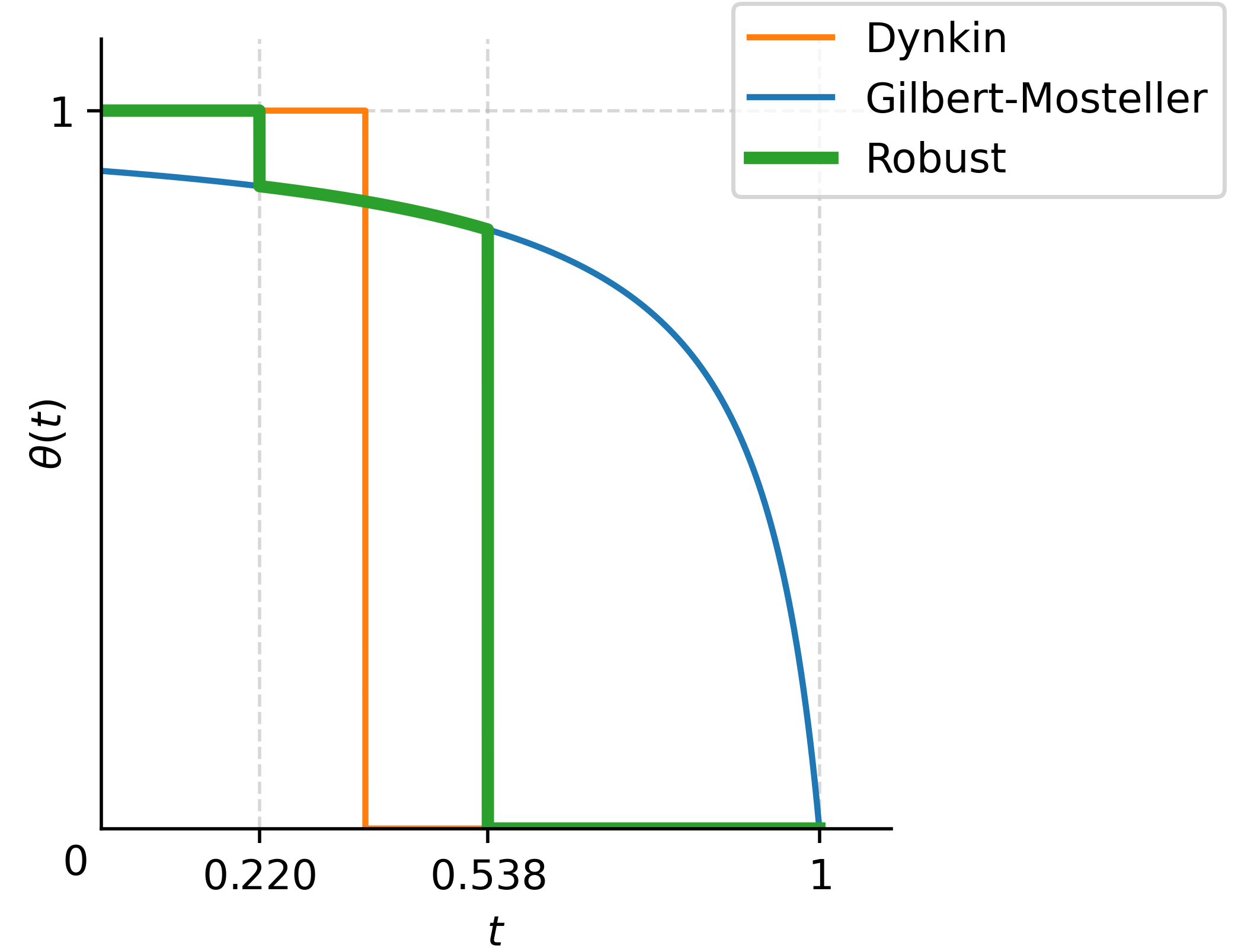}
    \vspace{-1em}
    \caption{Illustration of the robust threshold function corresponding to the Gilbert–Mosteller algorithm, with $n=10$, $\beta = \nicefrac{1}{3}$, $\lambda_1= \ee^{\W_{-1}(-\nicefrac{1}{3})} \approx 0.220$ and $\lambda_2= \ee^{\W_0(-\nicefrac{1}{3})}\approx 0.538$.}
    \label{fig:max-prob-threshold}
\end{figure}

\subsection{Optimal Robust Threshold Functions}
\label{subsec:optimal-robust-thresholds}

By~\Cref{lem:pre-robustness}, once $0\leq \lambda_1\leq \lambda_2\leq 1$ is fixed, every robust threshold of the form $\restrictsmall{\threshold}$ is guaranteed to be $\min\{-\lambda_1 \ln \lambda_1, -\lambda_2 \ln \lambda_2\}$-robust for both \MaxProb and \MaxExp.
Therefore, under this certified robustness lower bound, we can optimize consistency by freely choosing the threshold values on $(\lambda_1,\lambda_2)$.
Our main result in this subsection is that, for any $n$, there exist objective-specific optimal thresholds $\maxexpoptthres$ (\MaxExp) and $\maxproboptthres$ (\MaxProb), whose robust counterparts $\restrictsmall{\maxexpoptthres}$ and $\restrictsmall{\maxproboptthres}$ maximize consistency for every $0\leq \lambda_1\leq \lambda_2\leq 1$.

\begin{lemma}[Uniqueness and simultaneous optimality of threshold functions]
\label{lem:unique-inheritance-optimal-threshold}
    For any $n$, any prior distribution $F$ and any objective, either \MaxExp or \MaxProb, there exists a unique non-increasing threshold function, denoted by $\maxexpoptthres$ for \MaxExp and by $\maxproboptthres$ for \MaxProb, such that, for any $0\leq \lambda_1\leq \lambda_2\leq 1$, the corresponding bi-criteria algorithm induced by the robust threshold $\restrictsmall{\maxexpoptthres}$ (for \MaxExp) or $\restrictsmall{\maxproboptthres}$ (for \MaxProb) achieves the maximum consistency among all algorithms $\algowiththres(F,\restrictsmall{\threshold})$, where $\threshold$ ranges over all threshold functions.
    
    Here, the function $\maxproboptthres$ is precisely the threshold function used in Example~\ref{example:gilbert-mosteller}.
\end{lemma}

We now sketch the intuition behind why the bi-criteria algorithm admits simultaneously optimal threshold functions.
The proof of \Cref{lem:unique-inheritance-optimal-threshold} relies on the backward-induction equations and on the existence, uniqueness, and monotonicity of the corresponding optimal thresholds for \MaxExp and \MaxProb; these technical ingredients are developed in \Cref {app:unique-inheritance-optimal-threshold}.

\begin{proof}[Proof Sketch]

For a fixed objective and prior $F$, consider the dynamic programming formulation of the bi-criteria algorithm with a (not yet specified) threshold function~$\threshold$. At each time $t$, the algorithm compares the immediate payoff from accepting the current value $x$ with the continuation payoff from rejecting and following the same threshold rule in the future. The optimal threshold at time $t$ is characterized by an indifference condition between these two options.

More precisely, suppose a value $x$ arrives at time $t$ and is a record.
If $x$ is right at the threshold, i.e., $F(x) = \threshold(t)$, then the algorithm rejects $x$, and we have the following.
\begin{itemize}
    \item \textbf{Prior-free criterion.} All future values smaller than $x$ will be rejected, because they cannot be a record.
    \item \textbf{Prior-dependent criterion.} Since $\threshold$ is non-increasing in $t$ and $\threshold(t) = F(x)$, the first future value larger than $x$ will both exceed the threshold and be a new record.
\end{itemize}

In sum, the algorithm will accept the first future value that exceeds $x$. 
The continuation payoff is therefore the expected value (for \MaxExp) or success probability (for \MaxProb) associated with that future acceptance, which depends only on $t$, $x$, and the prior $F$, and \emph{does not} depend on the robustness parameters $(\lambda_1,\lambda_2)$.

In \Cref{app:unique-inheritance-optimal-threshold}, we derive that for each objective and each time $0 \le t\le 1$, the indifference (Bellman) equation that characterizes the optimal threshold $\maxexpoptthres(t)$ for \MaxExp and $\maxproboptthres(t)$ for \MaxProb.
There we show that each equation has a unique solution for every $t$, and that the resulting threshold function is non-increasing in $t$.
The detailed arguments appear in \Cref{lem:indifference-existence-uniqueness-MaxExp,lem:indifference-existence-uniqueness-MaxProb} in the appendix.

Since the continuation payoffs are independent of $(\lambda_1,\lambda_2)$ when the predicted prior is correct, the same objective-specific threshold solves the indifference equation of time $t$ for every choice of $\lambda_1 \leq t \leq \lambda_2$, and for any $0\leq \lambda_1\leq \lambda_2 \leq 1$, the robust threshold $\restrictsmall{\maxexpoptthres}$ (for \MaxExp) or $\restrictsmall{\maxproboptthres}$ (for \MaxProb) maximizes consistency among all $\restrictsmall{\threshold}$ in the bi-criteria framework.

This concludes our proof.
\end{proof}

For the \MaxProb objective, setting the robustness requirement to $\beta=0$ makes the predicted-prior problem coincide with the classical i.i.d.\ prophet inequality with the \MaxProb objective.
By \Cref{lem:unique-inheritance-optimal-threshold}, the threshold $\maxproboptthres$ is the unique consistency-maximizing threshold for this objective within our continuous-time framework.
As promised in \Cref{example:gilbert-mosteller}, we show that it achieves the optimal $0.580$ competitive ratio in \Cref{sec:algo-MaxProb}, where we prove \Cref{thm:max-prob-consistency-robustness}.

In contrast, for \MaxExp, our optimal threshold $\maxexpoptthres$ is not directly comparable to those of optimal i.i.d.\ prophet inequality algorithms, which rely solely on a threshold-based acceptance criterion without any record-based condition.
It remains open whether one can attain the $0.745$ worst-case consistency guarantee within this structural class, or whether any such record-based threshold policy must fall strictly below this benchmark.

\subsection{Brief Discussion on Algorithms Outside the Framework}

In terms of \MaxExp, the literature has considered algorithms that accept the first value greater than a threshold at its arrival time, regardless of whether the value is a record.
With a correct prior, the optimal policy by backward induction takes this form, where the threshold is the expected accepted value if the algorithm rejects the current value. 
For example, see the work by \citet{krengel1977semiamarts} on general prophet inequality and the more recent research~\cite{correaPostedPriceMechanisms2017,liuVariableDecompositionProphet2021,harbNewProphetInequalities2025} for i.i.d.\ prophet inequality.
These algorithms are outside our framework, where the taken value must be a record, in addition to surpassing the time-dependent threshold.

We stress that record-based requirement is critical for robustness.
The mentioned algorithms without record-based requirements are only $0$-robust. 
To see this, consider a predicted prior supported on values much smaller than those in the real prior.
Then, all values from the real prior will have a cdf value essentially $1$ w.r.t.\ the predicted prior.
As a result, the algorithm accepts the first value that arrives when $\threshold(t) < 1$, which is noncompetitive and leads to $0$ robustness.

Meanwhile, our framework is limited in that it does not adjust thresholds adaptively based on the observed values. 
Intuitively, a ``learning'' algorithm may adjust thresholds toward Dynkin's algorithm~\cite{dynkinOptimumChoiceInstant1963} if the observed values are unlikely to be from the predicted prior, and conversely, toward the $0.580$-competitive algorithm~\cite{gilbertRecognizingMaximumSequence} in the \MaxProb setting or the $0.745$-competitive algorithm~\cite{correaProphetInequalitiesIID2019} in the \MaxExp setting if the observed values are consistent with the predicted prior.
We consider a static threshold function to make the analyses tractable.
Exploring algorithms with adaptive thresholds is an interesting direction for future work.

\section{Implicit Sharding}
\label{sec:sharding}

This section introduces an implicit sharding technique that reduces instances with a smaller number of values to instances with arbitrarily many values.
In this sense, it suffices to design algorithms for the limit case when the number of values $n$ tends to infinity.
The main result of this section is the following theorem.

\begin{theorem}
    \label{thm:max-exp-sharding}
    For any $n,k\in \mathbb{N}$, for any predicted prior $\prior$, and for both the \MaxExp and \MaxProb objectives, if there is a threshold function $\threshold$ such that algorithm $\algowiththres\big(\prior,\threshold\big)$ is $\alpha$-consistent and $\beta$-robust for $nk$ values, then there is a corresponding $\alpha$-consistent $\beta$-robust algorithm $\mathcal{A}_{\mathrm{sharding}}(\tilde{F}, \threshold)$ for $n$ values.
\end{theorem}

By convention, $\tilde{F}^{\nicefrac{1}{k}}$ denotes the distribution defined by cdf $\tilde{F}^{\nicefrac{1}{k}}(\cdot)$.
Importantly, the maximum of $k$ values drawn i.i.d.\ from $\tilde{F}^{\nicefrac{1}{k}}$ follows distribution $\tilde{F}$.
We will construct an algorithm $\mathcal{A}_{\mathrm{sharding}}(\tilde{F}, \threshold)$ for $n$ values $x_1, x_2, \ldots, x_n$ to mimic the decision of algorithm $\mathcal{A}(\tilde{F}^{\nicefrac{1}{k}}, \threshold)$ with $nk$ values $x'_{11},\ldots,x'_{1k}, \ldots, x'_{n1},\ldots,x'_{nk}$ drawn i.i.d.\ from $\tilde{F}^{\nicefrac{1}{k}}$.

Na\"{i}vely, we would like to simulate algorithm $\mathcal{A}(\tilde{F}^{\nicefrac{1}{k}}, \threshold)$ by sampling $x'_{ij}$, $j \in [k]$, from $F^{\nicefrac{1}{k}}$ conditioned on the realization of $x_i$, and accept value $x_i$ if and only if one of the $x'_{ij}$'s is accepted by $\mathcal{A}(\tilde{F}^{\nicefrac{1}{k}}, \threshold)$.
This is the sharding technique by \citet{harbNewProphetInequalities2025} for the prophet inequality setting.
Unfortunately, the algorithm does not know $F$ in the problem of optimal stopping with a predicted prior, and thus cannot sample $x'_{ij}$ from $F^{\nicefrac{1}{k}}$.
Sampling $x'_{ij}$ based on the predicted prior $\tilde{F}$ and the corresponding $\tilde{F}^{\nicefrac{1}{k}}$ preserves the consistency of $\mathcal{A}(\tilde{F}^{\nicefrac{1}{k}}, \threshold)$, but not robustness.

We propose an implicit sharding method that partially simulates algorithm $\mathcal{A}(\tilde{F}^{\nicefrac{1}{k}}, \threshold)$ only at the arrival times of the $x_i = \max_{j \in [k]} x'_{ij}$.
We will show that this partial simulation is at least as good as the full simulation in the original sharding method.
In contrast to the sharding technique, the partial simulation does not rely on $F$ and will not explicitly construct the $x'_{ij}$'s, as they are only used in the analysis. Our algorithm is defined as follows.

\begin{tcolorbox}[title={Bi-criteria Algorithm with Implicit Sharding $\mathcal{A}_{\mathrm{sharding}}(\tilde{F}, \threshold)$}]
    \label{algo:algo-framework-new}
    \begin{enumerate}
        \item Sample $nk$ arrival times i.i.d.\ from $\mathrm{Unif}[0,1]$ and sort them as:\\[1ex]
        \hspace*{1.5in} 
        $t'_{11} < \cdots < t'_{1k} < \cdots < t'_{n1} < \cdots < t'_{nk}$~.
        \item For each $i \in [n]$, sample $s_i$ uniformly from $\{t'_{i1}, t'_{i2}, \ldots, t'_{ik}\}$.
        \item For each value $x_i$ when no value has been accepted, accept $x_i$ if
        \begin{enumerate}
            \item Value $x_i$ is a record; and
            \item The cdf value of $x_i$ w.r.t.\ $\tilde{F}^{\nicefrac{1}{k}}$ is greater than the threshold at time $s_i$, i.e., if\\[1ex]
            \hspace*{1.8in}
            $\prior^{\nicefrac{1}{k}}(x_i)> \threshold(s_{i}) ~.$
        \end{enumerate}
    \end{enumerate}
\end{tcolorbox}

Next, we design a coupling that gives a joint distribution preserving the marginal distributions of the random variables used in our algorithm $\mathcal{A}_{\mathrm{sharding}}(\tilde{F}, \threshold)$ and in the simulated algorithm $\algowiththres(\prior^{\nicefrac{1}{k}},\threshold)$, respectively.

\begin{definition}[Coupling]
\label{def:coupling-sharding}
Sample $nk$ i.i.d.\ values from $F^{\nicefrac{1}{k}}$. For each value, sample an arrival time i.i.d.\ from $\mathrm{Unif}[0,1]$. 
Then, we index the values by the ascending order of their arrival times. Denote the values by $\{x'_{ij}\}_{i\in[n],j\in[k]}$ and the sorted times by
\[
    t'_{11} < \cdots < t'_{1k} < \cdots < t'_{n1} < \cdots < t'_{nk}.
\]
For each $i$, define
\[
    x_i \coloneqq \max_{j\in[k]} x'_{ij},
    \qquad
    s_i \coloneqq t'_{i j_i},
    \quad \text{where } j_i\in\argmax_{j\in[k]} x'_{ij}.
\]
\end{definition}

Under this coupling, each $x_i$ has distribution $F$, and $\max_{i\in[n]} x_i = \max_{i\in[n],j\in[k]} x'_{ij}$ pointwise.
Also, because the shards in each block are i.i.d., the argmax index $j_i$ is uniform over $[k]$, so $s_i$ has the same distribution as Step 2 in $\mathcal{A}_{\mathrm{sharding}}$.

We stress that the algorithm does not rely on $F$ and will not explicitly construct the $x'_{ij}$'s. They are only used in the coupling for analysis.

\begin{proof}[Proof of \Cref{thm:max-exp-sharding}]
    We establish the \emph{pointwise value dominance} property: for any prior $F$, and any realization of the values and arrival times \emph{under the defined coupling}, the value accepted by algorithm $\mathcal{A}_{\mathrm{sharding}}(\tilde{F}, \theta)$ is greater than or equal to the value accepted by algorithm $\algowiththres(\prior^{\nicefrac{1}{k}},\threshold)$.

    Suppose algorithm $\mathcal{A}_{\mathrm{sharding}}(\tilde{F}, \theta)$ accepts value $x_i$, and $x'_{ij} = x_i$ is the corresponding value in the coupled instance for algorithm $\mathcal{A}(\tilde{F}^{\nicefrac{1}{k}}, \theta)$.
    Since $x_i$ is a record, $x'_{ij}$ is also a record because
    \[
        x'_{ij} = x_i \ge x'_{i'j'}, t'_{ij}=s_i \ge t'_{i'j'} \quad \text{for any } i' < i,\; j' \in [k]\text{ or } i'=i, \; j'\leq j.
    \]
    Further, the cdf value of $x'_{ij}$ is greater than the threshold at time $t'_{ij}$, as
    \[
        \tilde{F}^{\nicefrac{1}{k}}(x'_{ij}) = \tilde{F}^{\nicefrac{1}{k}}(x_i) > \theta(s_i) = \theta(t'_{ij})
        ~.
    \]
    Hence, algorithm $\mathcal{A}(\tilde{F}^{\nicefrac{1}{k}}, \theta)$ accepts either $x'_{ij} = x_i$, or a value that arrives earlier, which must be at most $x'_{ij} = x_i$ since $x_i$ is a record.
    
    By this pointwise value dominance, we obtain the following conclusion for each objective:
    \begin{enumerate}
        \item For \MaxExp: The conclusion follows directly by taking expectations over the coupling.
        
        \item For \MaxProb: The conclusion follows directly by taking expectations over all points in the coupling where $\mathcal{A}(\tilde{F}^{\nicefrac{1}{k}}, \theta)$ accepts the maximum value.
        
    \end{enumerate}
    
    This completes the proof.
\end{proof}

\section{Proof of Robustness}
\label{sec:robustness}

In \Cref{sec:algo} we stated Lemma~\ref{lem:pre-robustness}, which
bounds the robustness of $\algowiththres(\prior,\threshold)$ in our stochastic model. We prove it in this section.
The argument proceeds via a stronger ``fixed values, random arrivals'' model, also known as the \emph{game of Googol}~\cite{fergusonWhoSolvedSecretary1989,gnedin1994solution}, in which the values $x_{1}, x_{2}, \ldots, x_{n}$ are treated as deterministic and only the arrival order is random.

The key step is the following bound in the game-of-Googol setting.

\begin{proposition}
    \label{prop:robustness}
    For any predicted prior $\prior$, any threshold function $\threshold$, any $\lambda_1, \lambda_2$ with $0 \le \lambda_1 \le \lambda_2 \le 1$, and any realized values $x_1, x_2, \dots, x_n$, algorithm
    $\algowiththres(\prior, \restrictsmall{\threshold})$ accepts the maximum value with probability at least
    \[
        \min\{-\lambda_1 \ln \lambda_1, -\lambda_2 \ln \lambda_2\}
        ~,
    \]
    over the random arrivals of the values.
\end{proposition}

We now derive Lemma~\ref{lem:pre-robustness} from Proposition~\ref{prop:robustness}.

\begin{proof}[Proof of \Cref{lem:pre-robustness}]
    Fix a predicted prior $\prior$ and a threshold function $\threshold$. By definition, for $\lambda_1 = \sup \{ 0 \le t \le 1 : \threshold(t) = 1 \}$ and $\lambda_2 = \inf \{ 0 \le t \le 1 : \threshold(t) = 0 \}$, we have $\threshold=\restrictsmall{\threshold}$. Thus, it suffices to show that $\algowiththres(\prior, \restrictsmall{\threshold})$ is $\min\{-\lambda_1 \ln \lambda_1, -\lambda_2 \ln \lambda_2\}$-robust.
    
    By the law of total probability, we have
    \[
        \Pr[\ALG = \OPT ] 
        ~=~ \E_{x_1, x_2, \ldots, x_n} \big[ \Pr_{t_1, t_2, \ldots, t_n} \big[\, \ALG = \OPT \mid x_1, x_2, \ldots, x_n \,\big] \big] 
        ~.
    \]
    Applying Proposition~\ref{prop:robustness}, the inner probability is at least $\min\{-\lambda_1 \ln \lambda_1, -\lambda_2 \ln \lambda_2\}$ for every realization, so the expectation is at least the same quantity. This proves robustness for \MaxProb.
    Moreover, it implies $\E [\ALG] \ge \min\{-\lambda_1 \ln \lambda_1, -\lambda_2 \ln \lambda_2\} \cdot \E[\OPT]$, since for any realization of the values, the accepted value equals the maximum value with the same probability as in the \MaxProb case. This proves robustness for \MaxExp as well.
\end{proof}

Therefore, it remains to prove Proposition~\ref{prop:robustness}.
We first derive an explicit expression for the probability that algorithm $\algowiththres(\prior,\threshold)$ accepts the maximum value.
Without loss of generality (by relabeling the distinct realized values), assume $x_1 < x_2 < \dots < x_n$, so $\OPT = x_n$.

\begin{lemma}
\label{lem:googol-prob}
    For any $n\in \mathbb{N}$, any predicted prior $\prior$, any threshold function $\threshold$, any $\lambda_1, \lambda_2$ with $0 \le \lambda_1 \le \lambda_2 \le 1$, and any realized values $x_1 < x_2 < \dots < x_n$, algorithm $\algowiththres\big(\prior, \restrictsmall{\threshold}\big)$ accepts the maximum value $x_n$ with probability
    \begin{equation}
        \label{eqn:robust}
        \int_{0}^{1}\Ind{\prior(x_n) > \threshold(t)}
        \bigg((1-t)^{n-1} + \sum_{i=1}^{n-1} (1-t)^{n-1-i}
        \int_{0}^{t}\Ind{\prior(x_i) \leq \threshold(s)} \dd s \bigg) \dd t
        ~.
    \end{equation}
\end{lemma}

\begin{proof}
    Since $x_n$ is always a record, the algorithm accepts $x_n$ if and only if 1) its cdf value under $\prior$ is greater than the threshold, i.e., $\prior(x_n) > \threshold(t_n)$, and 2) the algorithm accepts no value before $t_n$.
    By~\Cref{prop:no-acceptance}, the latter is equivalent to the event that the prefix maximum for time $t_n$ has cdf under $\prior$ at most the threshold at its arrival time.
    Hence, the acceptance probability is
    \[
        \int_{0}^{1} \Ind{\prior(x_n) > \threshold(t_n)} \Pr \big[\, \text{prefix maximum for time $t_n$ is below threshold} \,\big] \dd t_n~.
    \]
    
    To compute the expression for the probability in the integrand, consider each possible prefix maximum separately.
    First, the prefix maximum is $0$ if no value arrives before $t_n$, which occurs with probability $(1-t_n)^{n-1}$.
    Further, the prefix maximum is $x_i$ if $x_i$ arrives at time $t_i \in [0, t_n)$, and the larger values $x_{i+1}, x_{i+2}, \ldots, x_{n-1}$ arrive after time $t_n$, which happens with probability $(1-t_n)^{n-1-i}$.
    In this case, we also require $\prior(x_i)\le \threshold(t_i)$ for the arrival time $t_i\in[0,t_n)$, contributing
    \[
        (1-t_n)^{n-1-i}
            \int_{0}^{t_n} \Ind{\prior(x_i) \leq \threshold(t_i)} \dd t_i
            ~.
    \]
    Putting everything together and changing variables $t_n$ to $t$ and $t_i$ to $s$, the lemma follows.
\end{proof}

We are now ready to prove \Cref{prop:robustness}.

\begin{proof}[Proof of \Cref{prop:robustness}]
By \Cref{lem:googol-prob}, it suffices to show that the expression in \Cref{eqn:robust} is at least $\min \{ - \lambda_1 \ln \lambda_1, - \lambda_2 \ln \lambda_2 \}$ for any predicted prior $\prior$ and threshold function $\threshold$.

First, we fix $\threshold$ and characterize the worst-case predicted prior $\prior$.
\Cref{eqn:robust} depends on $\prior$ through $n$ values $0 \le \prior(x_1) \le \prior(x_2) \le \cdots \le \prior(x_n) \le 1$.
Further, it is non-decreasing in $\prior(x_n)$, and non-increasing in $\prior(x_1), \prior(x_2), \ldots, \prior(x_{n-1})$.
Therefore, given any fixed threshold function $\threshold$, \Cref{eqn:robust} is minimized when all $\prior(x_i)$ are equal to some $q \in [0, 1]$.

In this case, the threshold function is effectively binary. 
A threshold $\threshold(t) < q$ is equivalent to $\threshold(t) = 0$, as $\prior(x_1), \prior(x_2), \ldots, \prior(x_n)$ all surpass the threshold. 
Similarly, a threshold $\threshold(t) \ge q$ is equivalent to $\threshold(t) = 1$, as none surpasses it.
Hence, we may without loss of generality consider a binary threshold, such that for some $\lambda \in [\lambda_1, \lambda_2]$, it holds
\[
    \threshold(t) = 
    \begin{cases}
        1, & \text{if } t \le \lambda~; \\
        0, & \text{otherwise~.}
    \end{cases}
\]

This corresponds to Dynkin's algorithm with a suboptimal transition point of $\lambda$ instead of the optimal $\nicefrac{1}{\ee}$.
This algorithm is $(- \lambda \ln \lambda)$-competitive (e.g., \citet{feldmanImprovedCompetitiveRatios2011}).
Since $- \lambda \ln \lambda$ is concave in $\lambda$ on $[0,1]$, its minimum over $[\lambda_1,\lambda_2]$ is attained at an endpoint, yielding the bound $\min\{-\lambda_1 \ln \lambda_1, -\lambda_2 \ln \lambda_2\}$.
\end{proof}

\section{Algorithm for \MaxExp: Proof of \Cref{thm:max-exp-consistency-robustness}}
\label{sec:algo-max-expect}

In this section we prove~\Cref{thm:max-exp-consistency-robustness}. 
We fix a $\beta \in [0,\nicefrac{1}{e}]$, and fix $0\le \lambda_1 \le \lambda_2\le 1$ as the two roots of $- \lambda \ln \lambda = \beta$.
Once there exists a non-increasing function $\bar{\threshold} \colon [\lambda_1,\lambda_2] \to [0, 1]$ that satisfies Equation~\eqref{eqn:thres-max-expect}, with a little abuse of notation, we extend the definition of $\bar{\threshold}$ to the entire interval $[0,1]$ by setting $\bar{\threshold}(t)=1$ for $t<\lambda_1$ and $\bar{\threshold}(t)=0$ for $t> \lambda_2$. 

By~\Cref{thm:max-exp-sharding}, it suffices to construct an algorithm that is both $\alpha$-consistent and $\beta$-robust as the number of values $n$ tends to infinity.
We show that $\algowiththres(\prior,\bar{\threshold}^{\nicefrac{1}{n}})$ satisfies these guarantees.
The $\beta$-robustness directly follows from \Cref{lem:pre-robustness}, since $\bar{\threshold}^{\nicefrac{1}{n}}=\restrictsmall{\bar{\threshold}^{\nicefrac{1}{n}}}$ by the extension above.

The rest of the section focuses on the consistency guarantee.
Since the predicted prior is correct, i.e., $\tilde{F} = F$, we will write it as $F$ throughout the section.
The remaining task is to show that $\algowiththres(F,\bar{\threshold}^{\nicefrac{1}{n}})$ is $\alpha$-competitive as $n \to \infty$.

As a standard approach in competitive analysis, we show that for any $\ell \ge 0$, it holds
\begin{equation}
    \label{eqn:consistency-per-quantile}
    \Pr \big[\, \ALG \ge \ell \,\big] 
    ~\ge~ 
    \alpha \cdot \Pr \big[\, \OPT \ge \ell \,\big]
    ~=~ \alpha \cdot \big( 1 - F(\ell)^n \big)
    ~,
\end{equation}
i.e., $\ALG$ stochastically dominates $\alpha \cdot \OPT$. This implies the desired consistency
\[
    \E \big[\, \ALG \,\big] 
    ~=~ 
    \int_0^\infty \Pr \big[\, \ALG \ge \ell \,\big] \dd \ell
    ~\ge~
    \alpha \cdot \int_0^\infty \Pr \big[\, \OPT \ge \ell \,\big] \dd \ell
    ~=~
    \alpha \cdot \E \big[\, \OPT \,\big] 
    ~.
\]

\begin{remark}
    The algorithm designed in this section does not use the theoretically optimal threshold $\restrictsmall{\maxexpoptthres}$ derived in \Cref{subsec:optimal-robust-thresholds};
    that threshold maximizes consistency among all algorithms with robust thresholds of the form $\restrictsmall{\threshold}$.
    For the \MaxExp objective, $\maxexpoptthres$ is given by $\maxexpoptthres(t)=F(x^*(t))$, where $x^*(t)$ is the unique solution of the Bellman equation (see \Cref{eq:max-exp-back-induction-app}), so $\maxexpoptthres$ depends on the prior distribution $F$ through quantities such as the conditional expectation $\E_{y\sim F}[y\mid y\ge x]$ and the quantile $F(x)$.
    As a result, the consistency guarantees $\Pr[\ALG\ge\ell]$ under $\maxexpoptthres$ are expressed in terms of these prior-dependent objects.
    In this section we instead work with a suboptimal but analytically more tractable threshold that is independent of the prior $F$, and we analyze its consistency relative to $\maxexpoptthres$ for each fixed prior.
\end{remark}

The proof consists of two parts. First, we simplify \Cref{eqn:consistency-per-quantile} to give a sufficient condition that $\bar{\threshold}$ needs to satisfy. Then, we verify that $\bar{\threshold}$ satisfies these conditions.

\subsection{Sufficient Conditions for Infinitely Many Values}

First, to simplify the left-hand side of \Cref{eqn:consistency-per-quantile}, we use the following lemma to calculate the probability a bi-criteria algorithm  $\algowiththres(F,\threshold^{\nicefrac{1}{n}})$ accepts a value that is at least $\ell$.

\begin{lemma}
    \label{lem:max_exp-alg-prob}
    For any threshold function $\threshold$, any number of values $n$, and any $\ell \geq 0$, algorithm $\algowiththres\big(F, \threshold^{\nicefrac{1}{n}} \big)$ satisfies that
        \begin{equation*}
        \label{eq:max_exp-algo-prob}
        \Pr \big[\, \ALG \ge \ell \,\big]
        ~=~
            \int_{F(\ell)^n}^{1}\int_{\threshold^{-1}(q)}^{1}\int_{0}^{t}
            \frac{1}{t} \Big( 1 - t + t \cdot  \min \big\{ \threshold(s), q \big\}^{\nicefrac{1}{n}} \Big)^{n-1} q^{-\frac{n-1}{n}}
            \dd s \dd t \dd q
        ~,
        \end{equation*}
    where $\threshold^{-1}(x)\coloneqq \inf\{t\in[0,1]\colon \threshold(t) < x\}$ is the generalized inverse of $\threshold$, defined as the boundary above which $\threshold(\cdot)$ is strictly smaller than $x$.
\end{lemma}

\begin{proof}
Consider any value $x_i$ arriving at time $t_i$, with the prefix maximum for time $t_i$ being $y$ and arriving at time $s$.
By~\Cref{prop:acceptance-rule}, the algorithm $\algowiththres\big(F, \threshold^{\nicefrac{1}{n}} \big)$ accepts $x_i$ if and only if
\[
    F(x_i)> \threshold(t_i)^{\nicefrac{1}{n}}
    ~,\quad 
    x_i\geq y
    ~,\quad 
    F(y) \leq \threshold(s)^{\nicefrac{1}{n}}
    ~.
\]
We consider the probability that $x_i$ is at least $\ell$ and accepted by the algorithm, enumerating over all triples $(x_i, t_i, s)$.
First, value $x_i$ follows distribution $F$ and satisfies $x_i\ge \ell$.
Further, its arrival time $t_i$ is uniform over $[0,1]$, and satisfies $F(x_i) > \threshold(t_i)^{\nicefrac{1}{n}}$, or equivalently, $F(x_i)^n > \threshold(t_i)$.
Finally, given $x_i$ and $t_i$, the arrival time $s$ of the prefix maximum is uniformly distributed over $[0,t_i]$. 
We may artificially sample $s$ from this distribution even if no values arrive before $t_i$. 
In sum, value $x_i$ is at least $\ell$ and accepted by the algorithm with probability
\[
    \int_{x_i= \ell}^{x_i = \infty} \int_{\threshold^{-1}(F(x_i)^n)}^{1} \int_{0}^{t_i}
    \frac{1}{t_i}\,
    \Pr \big[\, y\leq x_i \,\,,\, F(y)\leq \threshold(s)^{\nicefrac{1}{n}} \mid x_i,t_i,s \,\big]
    \dd s \dd t_i \dd F(x_i)
    ~.
\]
For each of the $n-1$ values other than $x_i$, either it arrives after time $t_i$, in which case there is no constraint on the realized value, or it arrives before $t_i$, in which case the cdf of the realized value is at most $\min \{ \threshold(s)^{\nicefrac{1}{n}} , F(x_i) \}$ by the upper bounds for $y$.
Thus,
\[
    \Pr \big[\, y\leq x_i \,\,,\, F(y)\leq \threshold(s)^{\nicefrac{1}{n}} \mid x_i,t_i,s \,\big] 
    ~=~
    \big(\,  1-t_i +t_i \min \{\threshold(s)^{\nicefrac{1}{n}},F(x_i)\}\,\big)^{n-1}
    ~.
\]
Putting together and changing variables with $t=t_i$ and $q=F(x_i)^n$, the above probability is
\[
    \frac{1}{n} \int_{F(\ell)^n}^{1}\int_{\threshold^{-1}(q)}^{1}\int_{0}^{t}
    \frac{1}{t} \Big( 1 - t + t \cdot  \min \big\{ \threshold(s), q \big\}^{\nicefrac{1}{n}} \Big)^{n-1} q^{-\frac{n-1}{n}}
    \dd s \dd t \dd q
    ~.
\]
Multiplying by $n$ by symmetry across all $x_i$'s proves the lemma.
\end{proof}

Now, focus on algorithm $\algowiththres(F,\bar{\threshold}^{\nicefrac{1}{n}})$. For any $y \in [0, 1]$, if we let $n\to \infty$ while keeping the same $F(\ell)^n = y$, the desired Inequality~\eqref{eqn:consistency-per-quantile} becomes
\[
    \int_{y}^1 \int_{\bar{\threshold}^{-1}(q)}^1 \int_0^t \frac{1}{t} \cdot \frac{\min \{ \bar{\threshold}(s), q \}^{t}}{q} \dd s \dd t \dd q
    ~\ge~ 
    \alpha \cdot ( 1 - y )
    ~,
\]    
where we use the fact that $\big( 1 - t + t z^{\frac{1}{n}} \big)^{n-1} \to z^t$ as $n$ approaches infinity.
\footnote{Rigorously, suppose we have $\alpha^*$-consistency in the limit when $n\to \infty$ and consider any $\varepsilon > 0$.
By the uniform convergence of the stated limit for $q \in [\varepsilon, 1]$, we can apply implicit sharding to ensure \Cref{eqn:consistency-per-quantile} for $\alpha = \alpha^* ( 1- \varepsilon)$ and any $q = F(\ell)^n \in [\varepsilon, 1]$.
Finally, \Cref{eqn:consistency-per-quantile} holds for $\alpha^* (1 - \varepsilon)^2$ when $q = F(\ell)^n \in [0, \varepsilon]$ because the left-hand-side is decreasing in $q$, while the right-hand-side differs by at most the $1-\varepsilon$ factor.}

Define the function with support set $[0, 1]$ as
\[
    g(q) ~\coloneqq~ \int_{\bar{\threshold}^{-1}(q)}^1 \int_0^t \frac{1}{t} \cdot \frac{\min \big\{ \bar{\threshold}(s), q \big\}^{t}}{q} \dd s \dd t
    ~.
\]
Then, it suffices to show that for any $q \in [0, 1]$, it holds
\[
    g(q)
    ~\ge~ 
    \alpha
    ~.
\]

\subsection{Verifying the Sufficient Condition}

Recall the assumed properties of $\bar{\threshold}$ by Equation~\eqref{eqn:thres-max-expect}:
\begin{enumerate}[label=\textit{\arabic*)}]
    \item $\bar{\threshold}$ is left-continuous, which implies $\bar{\threshold}(\lambda_1) = 1$, and non-increasing; and \label{property:left-continuous}
    \item For any $z \in [\lambda_1, \lambda_2]$, we have
    \[
        \int_z^1 \int_0^t \frac{1}{t} \cdot  \bar{\threshold} ( \max \{s, z \} )^t \dd s \dd t ~\ge~ 
        \alpha \cdot \bar{\threshold}(z)
        ~.
    \]
    \label{property:lower-bound}
\end{enumerate}

Consider any $q$ with $\bar{\threshold}(\lambda_2) \le q \le \bar{\threshold}(\lambda_1) = 1$.
Let $z = \bar{\threshold}^{-1}(q)$, and then by property \ref{property:lower-bound}, we have
\begin{align*}
    \alpha 
    &~\le~ 
    \int_{\bar{\threshold}^{-1}(q)}^1 \int_0^t \frac{1}{t} \cdot \frac{\bar{\threshold} ( \max \{s, \bar{\threshold}^{-1}(q) \} )^t}{\bar{\threshold}(\bar{\threshold}^{-1}(q))} \dd s \dd t \\
    &~=~ 
    \int_{\bar{\threshold}^{-1}(q)}^1 \int_0^t \frac{1}{t} \cdot \frac{\min \{ \bar{\threshold}(s), \bar{\threshold}(\bar{\threshold}^{-1}(q)) \}^{t}}{\bar{\threshold}(\bar{\threshold}^{-1}(q))} \dd s \dd t~.
\end{align*}
By left-continuity of $\bar{\threshold}$ and the definition of $\bar{\threshold}^{-1}$, we have $\bar{\threshold}(\bar{\threshold}^{-1}(q)) \ge q$.
Therefore,
\[
    \alpha 
    ~\le~
    \int_{\bar{\threshold}^{-1}(q)}^1 \int_0^t \frac{1}{t} \cdot  \frac{\min \{ \bar{\threshold}(s), q \}^t}{q} \dd s \dd t
    ~=~ g(q)~.
\]
Finally, we verify $g(q) \ge \alpha$ for $q \in [0, \bar{\threshold}(\lambda_2))$ by reducing it to the first case.

\begin{lemma}
    \label{lem:max-exp-g-left}
    For any $q \in [0, \bar{\threshold}(\lambda_2)]$, we have
    \[
        g(q) \ge g\big(\bar{\threshold}(\lambda_2)\big)~.
    \]
\end{lemma}

\begin{proof}
    For any $q$ in the range $[0, \bar{\threshold}(\lambda_2)]$, we have $\bar{\threshold}^{-1}(q) = \lambda_2$.
    Also, observe that $\bar{\threshold}(s) \ge \bar{\threshold}(\lambda_2) \ge q$ for $s \in [0, \lambda_2]$, and $\bar{\threshold}(s) = 0$ for $s \in [\lambda_2, 1]$.
    Hence,
    \[
        g(q) = \int_{\lambda_2}^1 \int_0^{\lambda_2} \frac{1}{t} q^{-(1-t)} \dd s \dd t = \int_{\lambda_2}^1 \frac{\lambda_2}{t} q^{-(1-t)} \dd t~,
    \]
    which is decreasing in $q$, implying the desired inequality.
\end{proof}

\section{Algorithm for \MaxProb: Proof of \Cref{thm:max-prob-consistency-robustness}}
\label{sec:algo-MaxProb}

This section proves~\Cref{thm:max-prob-consistency-robustness} by establishing the optimal consistency-robustness trade-off for the \MaxProb objective.
Recall that $\maxproboptthres$ is the threshold function used by the continuous-time Gilbert-Mosteller algorithm (\Cref{example:gilbert-mosteller}), where for each $t \in [0,1]$, the value $\maxproboptthres(t)$ is the unique solution $q$ to
\begin{equation}
    \label{eqn:max-prob-threshold-recall}
    \int_t^1 \frac{(1 - s + s q)^ {n-1} - q^{n - 1}}{1 - s} \dd s = q^{n-1}
    ~.
\end{equation}

By the Implicit Sharding technique (\Cref{thm:max-exp-sharding}), for any $\beta \in [0,\nicefrac{1}{e}]$, we can focus on the trade-off of the bi-criteria algorithm $\algowiththres(\prior,\restrictsmall{\maxproboptthres})$ in the limiting case when $n \to \infty$, where $\lambda_1,\lambda_2$ are the two roots of $- \lambda \ln \lambda = \beta$ with $\lambda_1 \le \lambda_2$.

Then, we show that this trade-off is tight within the bi-criteria framework by showing that no algorithm in the form of $\algowiththres(\prior,\threshold)$ can achieve a better consistency-robustness trade-off.

\subsection{Asymptotic Performance Analysis}
With \Cref{lem:pre-robustness}, the $\beta$-robustness guarantee of $\algowiththres(\prior, \restrict{\maxproboptthres})$ follows directly. It remains to show that the algorithm is $\alpha$-consistent as $n\to \infty$, where
\begin{equation*}
        \alpha = \beta  + \int_{\lambda_1}^{\lambda_2} \int_s^1 \,
        \frac{\ee^{-\frac{\constint t}{1-s}}}{t} \dd t \dd s
    ~,
\end{equation*}
To this end, we first derive an explicit expression for the probability that a bi-criteria algorithm accepts the maximum value under a correct prior for finite $n$.

\begin{lemma}
    \label{lem:algo-win-prob}
    For any $n \in \mathbb{N}$, any prior $F$, and any threshold function $\threshold$, algorithm $\algowiththres(\prior = F, \threshold)$ accepts the maximum value with probability
    \begin{equation*}
        \Gamma_n(\threshold) \coloneq \int_{0}^{1} \bigg( \int_{s}^1 \frac{(1-t+t\threshold(s))^n - t\threshold(s)^n}{t(1-t)} \dd t \; -\;\threshold(s)^n \bigg) \dd s \,.
    \end{equation*}
\end{lemma}

\begin{proof}
    Consider any value $x_i$ and the probability that $x_i$ is both maximum and accepted by the algorithm.
    Recall that $t_i$ is the arrival time of $x_i$, and let $s$ denote the arrival time of the prefix maximum $y$ for time $t_i$.
    By~\Cref{prop:acceptance-rule}, the algorithm accepts $x_i$ if and only if
    \begin{equation*}
        F(x_i) > \threshold(t_i)
        ~,\quad 
        x_i \ge y
        ~,~\text{and}\quad
        F(y) \leq \threshold(s)
        ~.    
    \end{equation*}

    We enumerate over all possible triples $(x_i, t_i, s)$.
    First, arrival time $t_i$ distributes uniformly over $[0, 1]$.
    Further, value $x_i$ follows distribution $F$, and must satisfy $F(x_i) \geq \threshold(t_i)$ for $x_i$ to be accepted.
    Finally, given $x_i$ and $t_i$, the arrival time $s$ of the prefix maximum is uniformly distributed over $[0,t]$. 
    We may artificially sample $s$ from this distribution even if no values arrive before $t_i$.
    Therefore, the probability that $x_i$ is both maximum and accepted equals
    \[
        \int_{0}^{1} \int_{\threshold(t_i)}^1 \int_0^{t_i} \frac{1}{t_i} \, \Pr \big[\, y \le x_i \,\,,\, F(y) \le \threshold(s) \,\,,\, \OPT = x_i \mid x_i, t_i, s \,\big] \dd s \dd F(x_i) \dd t_i
        ~.
    \]

For each of the $n-1$ values other than $x_i$, either it arrives before time $t_i$, in which case its cdf value must be at most $\min \{ \threshold(s) , F(x_i) \}$ by the upper bounds for the prefix maximum $y$; or it arrives after $t_i$, in which case it must be at most $x_i$ since $x_i$ is maximum.
Therefore,
\[
    \Pr \big[\, y \le x_i \,\,,\, F(y) \le \threshold(s) \,\,,\, \OPT = x_i \mid x_i, t_i, s \,\big] 
    ~=~
    \Big(\, t_i \min \big\{\threshold(s),F(x_i)\big\} + (1-t_i) F(x_i) \,\Big)^{n-1}
    ~.
\]

Putting together, changing variables with $t = t_i$ and $q = F(x_i)$, and multiplying by $n$ by symmetry across all $x_i$'s, the algorithm accepts the maximum value with probability
\[
    n \int_{0}^{1} \int_{\threshold(t)}^1 \int_0^t \frac{1}{t}  \, \big(\, t \min \{\threshold(s), q\} + (1-t) q \,\big)^{n-1} \dd s \dd q \dd t
    ~.
\]
Changing the order of integration and simplifying yield the expression for $\Gamma_n(\threshold)$.
\end{proof}

With the probability guarantee for general threshold derived in \Cref{lem:algo-win-prob}, we next evaluate the consistency achieved by our choice of thresholds for any $\beta \in [0,\nicefrac{1}{\ee}]$.

\begin{lemma}
\label{lem:max-prob-consistency}
For any $\beta \in [0,\nicefrac{1}{e}]$, suppose $\lambda_1$, $\lambda_2$ are the two roots of $-\lambda \ln \lambda=\beta$ with $\lambda_1 \leq \lambda_2$.
For any $n \in \mathbb{N}$, recall that $\maxproboptthres(t)$ is the unique solution to \Cref{eqn:max-prob-threshold-recall} for any $t \in [0, 1]$. 
Then
\begin{equation*}
\lim_{n\to \infty} \Gamma_n \big(\restrict{\maxproboptthres} \big) ~=~ \beta+ \int_{\lambda_1}^{\lambda_2} \int_s^1 
\frac{\ee^{-\frac{\constint t}{1-s}}}{t} \dd t \dd s
~.
\end{equation*}
\end{lemma}

\begin{proof}

    The probability $\Gamma_n \big(\restrict{\maxproboptthres} \big)$ can be written as
    \[
        \int_0^{\lambda_1} \bigg( \int_s^1 \frac{1}{t} \dd t - 1 \bigg) \dd s + \int_{\lambda_1}^{\lambda_2} \bigg( \int_{s}^1 \frac{(1-t+t\maxproboptthres(s))^n - t\maxproboptthres(s)^n}{t(1-t)} \dd t -\maxproboptthres(s)^n \bigg) \dd s + \int_{\lambda_2}^1 \frac{(1-t)^{n-1}}{t} \dd s
        ~.
    \]

    When $n \to \infty$, the first term equals $- \lambda_1 \ln \lambda_1=\beta$ by basic calculus, and the last term tends to $0$.
    It remains to show that, as $n \to \infty$, the inner integrand of the second term converges to
    \[
        \int_s^1 \frac{\ee^{-\frac{\constint t}{1-s}}}{t} \dd t
        ~.
    \]

    By the definition of $\maxproboptthres$ in \Cref{eqn:max-prob-threshold-recall}, the integrand equals
    \begin{align*}
        & \quad \int_{s}^1 \frac{(1-t+t\maxproboptthres(s))^n - t\maxproboptthres(s)^n}{t(1-t)} \dd t 
        \, - \,
        \maxproboptthres(s) \int_{s}^1 \frac{\big( 1 - t + t\maxproboptthres(s) \big)^{n-1} - \maxproboptthres(s)^{n - 1}}{1 - t} \dd t \\
        &= \int_{s}^1 \frac{(1-t+t\maxproboptthres(s))^{n-1}}{t} \dd t
        ~.
    \end{align*}

    Finally, we use the asymptotic expansion of the threshold function proven in \Cref{app:algorithm-analysis-max-prob-consistency}:
    \[
        \maxproboptthres(s)
        ~=~ \frac{1}{1+\frac{\constint}{(n-1)(1-s)}} + o\!\left(\frac{1}{n}\right)
    \]
    which holds uniformly on compact subsets of $[0,1)$.
    Substituting this limit into the integral yields the lemma.\footnote{Rigorously, for any $\varepsilon>0$, uniform convergence over $[0,1-\varepsilon]$ implies the convergence of $\int_s^{1-\varepsilon} \frac{\ee^{-\frac{\constint t}{1-s}}}{t} \dd t$. The difference from the ideal term is bounded by $\int_{1-\varepsilon}^1 \frac{1}{t} \dd t = \bigO(\varepsilon)$, which vanishes as $\varepsilon \to 0$.}
\end{proof}

\subsection{Tightness of the Trade-off}
\label{subsec:tightness-max-prob}

To complete the proof of \Cref{thm:max-prob-consistency-robustness}, we show that no bi-criteria algorithm can achieve a better consistency-robustness trade-off.

First, we show that for any bi-criteria algorithm $\algowiththres(\prior,\threshold)$, the robustness guarantee given by~\Cref{lem:pre-robustness} is tight.

\begin{lemma}
\label{lem:robustness-tightness}
For any $n\in \mathbb{N}$, any threshold function $\threshold$, suppose $\lambda_1=\sup\{0\le t\le 1 \colon \threshold(t) = 1\}$, $\lambda_2=\inf\{0\le t\le 1\colon \threshold(t) = 0\}$.
Then, for any prior $F$ and any $\varepsilon > 0$, there exists a predicted prior $\prior$ such that, when values are drawn i.i.d.\ from $F$, algorithm $\algowiththres(\prior,\threshold)$ accepts a maximum value with probability at most
\[
    \min\{-\lambda_1 \ln \lambda_1, -\lambda_2 \ln \lambda_2\}+\varepsilon
    ~.
\]
\end{lemma}

\begin{proof}

    Fix any prior $F$ and $\varepsilon > 0$.
    We prove two bounds, corresponding to ``small'' and ``large'' predicted priors.

    \paragraph{Case 1: Small predicted prior.}
    Define $\prior_{\mathrm{small}}(x)=1$ for all $x\ge 0$ (equivalently, a point mass $1$ at $0$).
    By monotonicity of $\threshold$, for any $t>\lambda_1$ we have $\threshold(t)<1$, so the second acceptance criterion $\prior_{\mathrm{small}}(x_i)>\threshold(t_i)$ always holds.
    Hence, up to a measure-zero boundary at $t=\lambda_1$, the algorithm rejects all values before $\lambda_1$ and then accepts the first record, i.e., Dynkin's continuous-time variant with transition point $\lambda_1$.
    By the standard analysis of that variant (e.g., \citet{feldmanImprovedCompetitiveRatios2011}), its \MaxProb ratio equals $-\lambda_1\ln\lambda_1$.
    Hence
    \[
        \Pr\!\big[\ALG(\prior_{\mathrm{small}},\threshold)=\OPT\big]
        \le
        -\lambda_1\ln\lambda_1
        ~.
    \]

    \paragraph{Case 2: Large predicted prior.}
    Since $F(\infty)=1$, there exists $m\ge 0$ such that
    \[
        F(m)^n \ge 1-\varepsilon
        ~.
    \]
    Define
    \[
        \prior_{\mathrm{large}}(x)=
        \begin{cases}
            0, & x < m+1~;\\
            1, & x \ge m+1~.
        \end{cases}
    \]
    Let
    \[
        E \coloneqq \{X_1\le m, X_2\le m, \ldots, X_n\le m\}
        ~,
    \]
    so $\Pr[E]\ge 1-\varepsilon$.
    On event $E$, all realized samples satisfy $\prior_{\mathrm{large}}(X_i)=0$.
    Therefore, for any $t<\lambda_2$, since $\threshold(t)>0$, the second acceptance criterion fails and the algorithm cannot accept before $\lambda_2$.
    Consequently, conditioned on $E$, this policy cannot outperform the policy that starts accepting records at time $\lambda_2$ (Dynkin variant), so
    \[
        \Pr\!\big[\ALG(\prior_{\mathrm{large}},\threshold)=\OPT \,\big|\, E\big]
        \le
        -\lambda_2\ln\lambda_2
        ~.
    \]
    Consequently,
    \begin{align*}
        \Pr\!\big[\ALG(\prior_{\mathrm{large}},\threshold)=\OPT\big]
        \le
        \Pr[E]\cdot(-\lambda_2\ln\lambda_2)+\Pr[E^c]
        \le
        -\lambda_2\ln\lambda_2+\varepsilon
        ~.
    \end{align*}

    Finally, choose $\prior_{\mathrm{small}}$ if $-\lambda_1\ln\lambda_1\le -\lambda_2\ln\lambda_2$, and choose $\prior_{\mathrm{large}}$ otherwise.
    In both cases,
    \[
        \Pr[\ALG=\OPT]
        \le
        \min\{-\lambda_1 \ln \lambda_1, -\lambda_2 \ln \lambda_2\}+\varepsilon
        ~.
    \]
    If one insists on full-support continuous predicted priors, the same conclusion follows up to an arbitrarily small additive term by the smoothing argument in~\Cref{sec:prelim}.
\end{proof}

We now complete the proof of~\Cref{thm:max-prob-consistency-robustness}.
Fix any bi-criteria algorithm $\algowiththres(\prior,\threshold)$, and define
\[
    \lambda_1' = \sup\{0\le t\le 1 \colon \threshold(t) = 1\}~,
    \qquad
    \lambda_2' = \inf\{0\le t\le 1 \colon \threshold(t) = 0\}~,
\]
\[
    \beta \coloneqq \min\{-\lambda_1' \ln \lambda_1', -\lambda_2' \ln \lambda_2'\}~.
\]
By~\Cref{lem:robustness-tightness}, for every $\varepsilon>0$ there exists a misspecification under which this algorithm succeeds with probability at most $\beta+\varepsilon$, hence its robustness is at most $\beta$.

Now let $\lambda_1,\lambda_2$ be the two roots of $-\lambda\ln\lambda=\beta$. By the concavity of $-\lambda\ln\lambda$ over $[0,1]$, we have $\lambda_1'\le \lambda_1 \le \lambda_2 \le \lambda_2'$.
This means $\threshold=\restrict{\threshold}$.

By~\Cref{lem:unique-inheritance-optimal-threshold}, for these fixed $(\lambda_1,\lambda_2)$ and any $n$, the maximum consistency within the class $\algowiththres(\prior,\restrict{\threshold})$ is achieved by $\algowiththres(\prior,\restrict{\maxproboptthres})$.
Therefore, no $\beta$-robust bi-criteria algorithm can have consistency larger than $\Gamma_n(\restrict{\maxproboptthres})$.
Taking $n \to \infty$ and applying~\Cref{lem:max-prob-consistency}, this upper bound converges to
\[
    \alpha
    =
    \beta+\int_{\lambda_1}^{\lambda_2}\int_s^1 \frac{\ee^{-\frac{\constint t}{1-s}}}{t}\dd t \dd s~.
\]
Hence, the trade-off in~\Cref{thm:max-prob-consistency-robustness} is tight within the bi-criteria framework.

\section{Hardness for \MaxProb}
\label{sec:hardness-MaxProb}

This section proves that certain consistency-robustness trade-offs for \MaxProb are unattainable.
We consider the discrete-time model and discrete distributions.
In the presence of ties, we adopt the all-ties-win convention, i.e., the algorithm only needs to accept any value that equals $\OPT$.
This convention only strengthens the hardness result.

\begin{theorem}
    \label{thm:hardness-maxprob}
    If there is an $\alpha$-consistent $\beta$-robust algorithm for \MaxProb, then for any integers $n, K \in \mathbb{N}$ and any predicted prior $\tilde{F}$ over $[K]$, the following polytope $\mathcal{P}_{n, K, \tilde{F}}(\alpha, \beta)$ for variables $\{ y_{t, \ell} \}_{t \in [n], \, \ell \in [K]}$ is feasible:
    \begin{align}
        \forall \ell,\quad & y_{0,\ell} = 1;
        \label{eq:max-prob-LP-S-initial-condition} \\[1.5ex]
        \forall t,\forall \ell,\quad & 0 \leq y_{t,\ell} \leq 1;
        \label{eq:max-prob-LP-S-in-0-1} \\[1.5ex]
        \forall t,\forall \ell,\quad & \Delta^{t}(\ell) \, y_{t,\ell} - \Delta^{t-1}(\ell)\, \tilde{F}(\ell-1) \, y_{t-1,\ell} \geq 0 \label{eq:max-prob-LP-inequality-1}\\[1.5ex]
        \forall t,\forall \ell,\quad & \Delta^{t}(\ell) \, y_{t,\ell} - \Delta^{t-1}(\ell)\, \tilde{F}(\ell-1) \, y_{t-1,\ell}
        \leq \sum_{m \in [\ell]} \Delta^{t-1}(m)\, \tilde{f}(\ell) \, y_{t-1,m};
        \label{eq:max-prob-LP-inequality-2} \\
        & \sum _{t \in [n]} \sum_{\ell \in [K]} \tilde{F}(\ell)^{n-t} \Bigg(
        \sum_{m \in [\ell]} \Delta^{t-1}(m)\, \tilde{f}(\ell) \, y_{t-1,m} + \Delta^{t-1}(\ell)\, \tilde{F}(\ell-1) \, y_{t-1,\ell} - 
         \Delta^{t}(\ell) \, y_{t,\ell}
        \Bigg) \geq \alpha;
        \label{eq:max-prob-LP-alpha} \\
        \forall k ,\quad & \sum_{t \in [n]} \sum_{\ell \in [k]} \frac{\tilde{F}(\ell)^{n-t} }{ \tilde{F}(k)^{n}}\Bigg(
        \sum_{m \in [\ell]} \Delta^{t-1}(m)\, \tilde{f}(\ell) \, y_{t-1,m} + \Delta^{t-1}(\ell)\, \tilde{F}(\ell-1) \, y_{t-1,\ell} - 
         \Delta^{t}(\ell) \, y_{t,\ell}
        \Bigg) \geq \beta,
        \label{eq:max-prob-LP-beta}
    \end{align}
    where $\tilde{f}$ is the pmf of $\tilde{F}$, and we define $\Delta^{t}(\ell) \coloneqq \tilde{F}(\ell)^{t}-\tilde{F}(\ell-1)^{t}$ for any $t\in [n]$, with $\Delta^{0}(\ell) \coloneqq \Ind{\ell = 1}$ under the convention $0^0 = 0$.
\end{theorem}

\paragraph{Additional Notations.}
For any vector $\bm{x} = (x_1, x_2, \dots,x_n) \in [K]^n$ and any $i, j \in [n]$, we define $\bm{x}_{i:j} \coloneqq (x_i, x_{i+1}\dots, x_j)$.
Further define the maxima $\max \bm{x}_{i:j} = \max_{i \le t \le j} x_t$ and $\max \bm{x} = \max_{t \in [n]} x_t$.
If $i > j$, we adopt the convention $\max \bm{x}_{i:j} = 1$
consistent with the assumption that $x_t \ge 1$ for all $t$.

\subsection{Family of Candidate Prior Distributions}
\label{subsec:truncated-distribution-family}

Given the predicted prior $\tilde{F}$ over $[K]$, we consider a family of candidate prior distributions by conditioning on having values at most $k \in [K]$.

\begin{definition}
\label{def:truncated-distribution-family}
Let $\mathcal{F} = \{F_k\}_{k \in [K]}$ be a family of distributions such that $F_k$ is the conditional distribution obtained by restricting $\tilde{F}$ to support $[k]$.
In other words, the cdf of $F_k$ is
\[
	F_k(\ell) 
	~=~ 
	\begin{cases}
	   \dfrac{\tilde{F}(\ell)}{\tilde{F}(k)}, & 1\le \ell \le k ~;\\[2ex]
	   ~ 1, & k < \ell \le K ~.
	\end{cases}
\]
\end{definition}

At the two extremes, we have $F_K = \tilde{F}$ and $F_1$ being a point mass at $1$.
We use $\Pr_k$ to denote probability under distribution $F_k$, and $\Pr$ for probabilities of events independent of the specific choice of the underlying distribution, e.g., those over the randomness of the algorithm.

The main property of this construction is the following.
Given any realized values in the first $t$ steps, the only information we can conclude about the real prior $F_k$ is $k \ge \max \bm{x}_{1:t}$.
Formally, we have the next lemma.

\begin{lemma}
    \label{lem:hardness-posterior}
    For any $t \in [n]$, any $\ell \in [K]$, and any subset $S \subseteq [\ell]^t$ of $t$ realized values whose maximum equals $\ell$, the conditional probability:
    \[
        \Pr_k \big[ \bm{x}_{1:t} \in S \mid \max \bm{x}_{1:t} = \ell \big] 
  \]
  is the same for all $k$ with $\ell \le k \le K$.
\end{lemma}

\begin{proof}
	By Bayes' rule, we have
	\[
		\Pr_k \big[ \bm{x}_{1:t} \in S \mid \max \bm{x}_{1:t} = \ell \big] 
		~=~
		\frac{\Pr_k \big[ \bm{x}_{1:t} \in S \,,\, \max \bm{x}_{1:t} = \ell \big]}{\Pr_k \big[ \max \bm{x}_{1:t} = \ell \big]}
		~.
	\]
	
	Replacing $k$ with $K$ corresponds to multiplying both the numerator and denominator by $\tilde{F}^t (k)$ by \Cref{def:truncated-distribution-family}, resulting in the same ratio.
\end{proof}

\subsection{Representations of Algorithms}
\label{subsec:stopping-rules}

To derive the linear constraints in polytope $\mathcal{P}_{n, K, \tilde{F}}(\alpha, \beta)$, we consider two representations of online algorithms when the predicted prior is $\tilde{F}$ and the real prior is one of the conditional distributions in the above family $\mathcal{F}$.
The first one relies on the conditional acceptance probabilities, which appeared in several previous works (e.g., \citet{buchbinderSecretaryProblemsLinear2014, correaSampledrivenOptimalStopping2021}).

\begin{definition}[Conditional Acceptance Probabilities]
    Given any algorithm, and for any $t \in [n]$, define $\accept_t: [K]^t \to [0,1]$ as the probability of accepting the value at time $t$, conditioned on the realized values up to time $t$, and not accepting any value before $t$.
    That is
    \begin{equation*}
        \accept_t(\bm{x}_{1:t}) 
        ~\coloneqq~
        \Pr \big[\, \stoptime = t \mid \stoptime \geq t \,,\, \bm{x}_{1:t} = \bm{x}_{1:t} \,\big]
        ~.
    \end{equation*}
\end{definition}

The conditional acceptance probabilities fully characterize the algorithm, as we can write the probability of accepting the value at step $t$ as:
\[
    \Pr\big[\stoptime=t \mid \bm{x}\big]
    ~=~
    \accept_t(\bm{x}_{1:t}) \prod_{i\in [t-1]} \big(1-\accept_i(\bm{x}_{1:i}) \big) 
    ~.
\]

Next, we show that we may without loss of generality focus on a class of \emph{minor-oblivious} algorithms, whose decisions for whether accepting a record $x_t$ at time $t$ are independent of smaller values $\bm{x}_{1:t-1}$ in the past.

\begin{definition}[Minor-oblivious Algorithms]
    \label{def:minor-oblivious}
    An algorithm is called \emph{minor-oblivious} if $\accept_t(\bm{x}_{1:t}) = \accept_t(\bm{y}_{1:t})$ for any $\bm{x}_{1:t}$ and $\bm{y}_{1:t}$ satisfying  $x_t=y_t=\max\bm{x}_{1:t}=\max\bm{y}_{1:t}$. 
    We then write $\accept_t(\ell)$ for the common value of $\accept_t(\bm{x}_{1:t})$ for all $\bm{x}_{1:t}$ satisfying $\max \bm{x}_{1:t} = \ell$.
\end{definition}

Recall that based on the realization of $\bm{x}_{1:t}$, the only information we may conclude about the prior $F_k$ is $k \ge \max \bm{x}_{1:t}$.
Intuitively, there is no reason to make decisions based on the details of $\bm{x}_{1:t}$ other than a maximum $\max \bm{x}_{1:t}$.
In this sense, minor-oblivious algorithms are without loss of generality.
We formalize this characterization in the following proposition, which may be of independent interest.
The proof is deferred to \Cref{app:proof-minor-oblivious}. 

\begin{proposition}
\label{prop:minor-oblivious}
If an algorithm is $\alpha$-competitive when the prior is $F_K$ and $\beta$-competitive when the prior is $F_k$ for $k \in [K - 1]$, there is a minor-oblivious algorithm with the same guarantees.
\end{proposition}

Focusing on minor-oblivious algorithms, we next consider a new representation by the probability of rejecting the first $t$ values conditioned on a maximum among them.

\begin{definition}[Minor-oblivious Rejection Probabilities]
\label{def:survival-probabilities}
For any $t \in [n]$ and any $\ell\in[K]$, define the minor-oblivious rejection probability $\rejall_t : [K] \to [0, 1]$ as:
\[
    \rejall_t(\ell)
    ~\coloneqq~
    \Pr \big[\, \stoptime > t \mid \max \bm{x}_{1:t}=\ell \,\big]
    ~.
\]
We artificially define $\rejall_0(\ell) \coloneqq 1$ to simplify some expressions in the rest of the section.
\end{definition}

There is a one-to-one correspondence between the two representations (see \Cref{app:survival-stopping}). 
We state below the direction to be used in the proof of \Cref{thm:hardness-maxprob}.

\begin{lemma}
\label{lem:survival-to-stopping}
For any minor-oblivious rule, any $t\in[n]$, and any $\ell\in[K]$, we have
\begin{equation}
\accept_t(\ell) ~=~ 1 - \frac{\rejall_t(\ell) \cdot \Pr_K \big[ \max \bm{x}_{1:t} = \ell \,\big] - \rejall_{t-1}(\ell) \cdot \Pr_K \big[\, x_t < \max \bm{x}_{1:t-1} = \ell \,\big]}
{\sum_{m\in [\ell]} \rejall_{t-1}(m) \cdot \Pr_K \big[ \max \bm{x}_{1:t-1} = m, x_t = \ell \,\big]}
~.
\label{eq:Rt-def}
\end{equation}
\end{lemma}

\subsection{Proof of \Cref{thm:hardness-maxprob}}

Given any minor-oblivious, $\alpha$-consistent, $\beta$-robust algorithm, we may assume without loss of generality that it rejects all values that are not a record, since we consider \MaxProb.
We next verify that $y_{t,\ell} = \rejall_t(\ell)$ is feasible w.r.t.\ polytope $\mathcal{P}_{n, K, \tilde{F}}(\alpha, \beta)$.

\paragraph{Bounds for Rejection Probabilities.}
By the definition of the minor-oblivious rejection probabilities, and the artificial definition of $\rejall_0(\ell)$, we obtain Constraints~\eqref{eq:max-prob-LP-S-initial-condition} and \eqref{eq:max-prob-LP-S-in-0-1}:
\[
    \rejall_0(\ell) \,=\, 1
    ~,\qquad
    0 \,\le\, \rejall_t(\ell) \,\le\, 1 
    ~,~~ \forall t\in[n],\; \forall \ell\in[K]
    ~.
\]

\paragraph{Bounds for Acceptance Probabilities.}
By $0 \le \accept_t(\ell) \le 1$ and \Cref{eq:Rt-def}, we get that
\begin{align*}
0 
& 
~\le~ \rejall_t(\ell) \cdot \Pr_K[\max \bm{x}_{1:t} = \ell] - \rejall_{t-1}(\ell) \cdot \Pr_K[\max \bm{x}_{1:t-1} = \ell , x_{t} < \ell] \\
&
~\le~ \sum_{m\in [\ell]} \rejall_{t-1}(m) \cdot \Pr_K[\max \bm{x}_{1:t-1} = m, x_t = \ell]
~.
\end{align*}

These two inequalities correspond to Constraints~\eqref{eq:max-prob-LP-inequality-1} and \eqref{eq:max-prob-LP-inequality-2}, because
\begin{align*}
    \Pr_K[\max \bm{x}_{1:t} = \ell] 
    &
    ~=~ \Delta ^t(\ell)
    ~, \\[1ex]
    \Pr_K[\max \bm{x}_{1:t-1} = \ell , x_{t} < \ell]
    &
    ~=~ \Pr_K[\max \bm{x}_{1:t-1} = \ell] \cdot \Pr_K[x_{t} < \ell] \\
    &
    ~=~ \Delta^{t-1}(\ell) \tilde{F}(\ell-1)
    ~, \\[1ex]
    \Pr_K[\max \bm{x}_{1:t-1} = m, x_t = \ell]
    & 
    ~=~ \Pr_K[\max \bm{x}_{1:t-1} = m] \cdot \Pr_K[x_t = \ell] \\
    &
    ~=~ \Delta^{t-1}(m) \tilde{f}(\ell)
    ~.
\end{align*}

\paragraph{Competitive Ratios.}
Consider the probability of accepting a maximum value when the prior is $F_k$.
We write it as:
\begin{align*}
    \Pr_k [ \ALG = \OPT ]
    &
    ~=~ \sum _{t \in [n]} \sum_{\ell \in [k]} \Pr_k \big[\, \stoptime = t \,,\, \max \bm{x}_{1:t}  = \max \bm{x} = \ell \,\big] \\
    &
    ~=~ \sum _{t \in [n]} \sum_{\ell \in [k]} \Pr_k \big[\, \max \bm{x}_{1:t}  = \max \bm{x} = \ell \,\big] \Pr \big[\, \stoptime = t \mid \max \bm{x}_{1:t} = \ell \,\big]
    ~.
\end{align*}
We did not explicitly write the optimality of $x_t$, i.e., $x_t = \max \bm{x}$, since $\stoptime = t$ implies $x_t = \max \bm{x}_{1:t}$ by our assumption without loss of generality that the algorithm rejects all values that are not a record.

By the relation between $\tilde{F} = F_K$ and $F_k$, we can rewrite it as:
\begin{align*}
    &
    \frac{1}{\tilde{F}(k)^n} \sum _{t \in [n]} \sum_{\ell \in [k]} \Pr_K \big[\, \max \bm{x}_{1:t} = \max \bm{x} = \ell \,\big] \Pr \big[\, \stoptime = t \mid \max \bm{x}_{1:t} = \ell \,\big] \\
    & \qquad\qquad
    ~=~ \frac{1}{\tilde{F}(k)^n} \sum _{t \in [n]} \sum_{\ell \in [k]} \Pr_K \big[\, \stoptime = t \,,\, \max \bm{x}_{1:t} = \max \bm{x} = \ell \,\big] 
    ~.
\end{align*}

Further, the event $\stoptime = t$ only depends on $\bm{x}_{1:t}$, while the realization of $\bm{x}_{t+1:n}$ is independent to both $\bm{x}_{1:t}$ and whether $\stoptime = t$.
We can further transform the above expression into:
\begin{align}
    &
    \frac{1}{\tilde{F}^n(k)} \sum _{t \in [n]} \sum_{\ell \in [k]} \Pr_K \big[\, \stoptime = t \,,\, \max \bm{x}_{1:t} = \ell \,\big] \Pr_K \big[\, \max \bm{x}_{t+1:n} \le \ell \,\big] \notag \\
    & \qquad\qquad
    ~=~ 
    \sum _{t \in [n]} \sum_{\ell \in [k]} \frac{\tilde{F}(\ell)^{n-t}}{\tilde{F}(k)^n} \Pr_K \big[\, \stoptime = t \,,\, \max \bm{x}_{1:t} = \ell \,\big]
    ~.
    \label{eqn:max-prob-hardness-ratio-0}
\end{align}

The final step is to write the above probability as a linear expression of the minor-oblivious rejection probabilities $\rejall_t(\ell)$'s.
First, we write $\Pr_K \big[\, \stoptime = t \,,\, x_t = \max \bm{x}_{1:t} = \ell \,\big]$ as:
\begin{equation}
    \label{eqn:max-prob-hardness-ratio-1}
    \Pr_K \big[\, \stoptime \ge t \,,\, \max \bm{x}_{1:t} = \ell \,\big] \,-\, \Pr_K \big[\, \stoptime > t \,,\, \max \bm{x}_{1:t} = \ell \,\big]
    ~.
\end{equation}

The first part equals:
\begin{align}
    & \Pr \big[\, \stoptime \geq t \,,\, \max \bm{x}_{1:t-1} \leq \ell \,,\, x_{t} = \ell \,\big] + \Pr \big[\, \stoptime \geq t \,,\, \max \bm{x}_{1:t-1} = \ell \,,\, x_{t} < \ell \,\big] \notag \\[2ex]
    & \qquad
    ~=~
    \sum_{m \in [\ell]} \Pr[\max \bm{x}_{1:t-1}=m] \cdot \Pr[x_{t}=\ell] \cdot \Pr \big[ \stoptime > t-1 \mid \max \bm{x}_{1:t-1}=m \big] \notag \\
    & \qquad\qquad\qquad~
    + \Pr[\max \bm{x}_{1:t-1} = \ell] \cdot \Pr[x_{t} < \ell] \cdot \Pr \big[ \stoptime > t-1 \mid \max \bm{x}_{1:t-1}= \ell \big] \notag \\[2ex]
    & \qquad
    ~=~
    \sum_{m \in [\ell]} \Delta^{t-1}(m) \cdot \tilde{f}(\ell) \cdot \rejall_{t-1}(m) + \Delta^{t-1}(\ell) \cdot \tilde{F}(\ell-1) \cdot \rejall_{t-1}(\ell)
    ~.
    \label{eqn:max-prob-hardness-ratio-2}
\end{align}

The second part equals:
\begin{equation}
    \label{eqn:max-prob-hardness-ratio-3}
    \Pr_K \big[\, \max \bm{x}_{1:t} = \ell \,\big] \cdot \Pr_K \big[\, \stoptime > t \mid \max \bm{x}_{1:t} = \ell \,\big]
    ~=~
    \Delta^t(\ell) \cdot \rejall_t(\ell)
    ~.
\end{equation}

Combining \Cref{eqn:max-prob-hardness-ratio-0,eqn:max-prob-hardness-ratio-1,eqn:max-prob-hardness-ratio-2,eqn:max-prob-hardness-ratio-3}, we conclude that when the real prior is $F_k$, the competitive ratio is equal to
\begin{align*}
    \sum _{t \in [n]} \sum_{\ell \in [k]} \frac{\tilde{F}(\ell)^{n-t}}{\tilde{F}(k)^n} \bigg(
    &
    \sum_{m \in [\ell]} \Delta^{t-1}(m) \cdot \tilde{f}(\ell) \cdot \rejall_{t-1}(m) \\
    & \quad
    + \Delta^{t-1}(\ell) \cdot \tilde{F}(\ell-1) \cdot \rejall_{t-1}(\ell) - \Delta^t(\ell) \cdot \rejall_t(\ell) 
    \bigg).
\end{align*}
Since the algorithm is $\alpha$-consistent and $\beta$-robust, the above competitive ratio is at least $\alpha$ when $k = K$, and at least $\beta$ when $k \in [K - 1]$.
These correspond to Constraints~\eqref{eq:max-prob-LP-alpha} and \eqref{eq:max-prob-LP-beta}.

\subsection{Numerical Results}

For any $\lambda \in [0, 1]$, we solve a linear program with variables $\{ y_{t, \ell} \}_{t \in [n], \, \ell \in [K]} \cup \{\alpha, \beta\}$ to maximize:
\[
    \lambda \cdot \alpha + (1- \lambda) \cdot \beta~,
\]
subject to the linear constraints that define polytope $\mathcal{P}_{n, K, \, \tilde{F}}(\alpha, \beta)$.
The optimal objective value $\mathrm{LP}^*(\lambda)$ then implies a hardness result of the form:
\[
    \lambda \cdot \alpha + (1- \lambda) \cdot \beta ~\ge~ \mathrm{LP}^*(\lambda)
    ~.
\]

We set $n = 30$ and $K = 1024$, and consider a predicted prior $\tilde{F}$ defined by $\tilde{f}(k) \propto \frac{1}{k}$ for $k \in [K]$. 
We solve the linear programs for $\lambda = \frac{i}{100}$ for $0 \le i \le 100$ to get the curve shown in \Cref{fig:MaxProb}.
For details, see \href{https://github.com/billyldc/optimal-stopping-with-a-predicted-prior}{our code repository}.\footnote{\url{https://github.com/billyldc/optimal-stopping-with-a-predicted-prior}}

\section*{Acknowledgement}

We are grateful for the insightful discussions with Yiding Feng, Enze Sun, and Xiaowei Wu. We also wish to thank the participants of the Dagstuhl workshop 2025 “Online Algorithms Beyond Competitive Analysis” and the “Trends in Approximation and Online Algorithms (TAO)” workshop at FOCS 2025 for their valuable feedback.

\bibliography{reference}

@inproceedings{gatmiry2024bandit,
  title={Bandit algorithms for prophet inequality and pandora's box},
  author={Gatmiry, Khashayar and Kesselheim, Thomas and Singla, Sahil and Wang, Yifan},
  booktitle={Proceedings of the 35th Annual ACM-SIAM Symposium on Discrete Algorithms},
  pages={462--500},
  year={2024}
}

@article{jin2024sample,
  title={Sample complexity of posted pricing for a single item},
  author={Jin, Billy and Kesselheim, Thomas and Ma, Will and Singla, Sahil},
  journal={Advances in Neural Information Processing Systems},
  volume={37},
  pages={82296--82317},
  year={2024}
}

@inproceedings{guo2021generalizing,
  title={Generalizing complex hypotheses on product distributions: Auctions, prophet inequalities, and pandora’s problem},
  author={Guo, Chenghao and Huang, Zhiyi and Tang, Zhihao Gavin and Zhang, Xinzhi},
  booktitle={Conference on Learning Theory},
  pages={2248--2288},
  year={2021}
}

@inproceedings{abolhassaniBeating11Ordered2017,
  title = {Beating 1-1/e for Ordered Prophets},
  booktitle = {Proceedings of the 49th Annual ACM SIGACT Symposium on Theory of Computing},
  author = {Abolhassani, Melika and Ehsani, Soheil and Esfandiari, Hossein and HajiAghayi, MohammadTaghi and Kleinberg, Robert and Lucier, Brendan},
  year = {2017},
  pages = {61--71}
}

@article{antoniadisSecretaryOnlineMatching2020,
  title={Secretary and online matching problems with machine learned advice},
  author={Antoniadis, Antonios and Gouleakis, Themis and Kleer, Pieter and Kolev, Pavel},
  journal={Advances in Neural Information Processing Systems},
  volume={33},
  pages={7933--7944},
  year={2020}
}

@article{krengel1977semiamarts,
  title = {Semiamarts and Finite Values},
  author = {Krengel, Ulrich and Sucheston, Louis},
  year = {1977},
  journal = {Bulletin of the American Mathematical Society},
  volume = {83},
  pages = {745--747},
  publisher = {American Mathematical Society}
}

@article{krengel1978semiamarts,
  title={On semiamarts, amarts, and processes with finite value},
  author={Krengel, Ulrich and Sucheston, Louis},
  journal={Advances in Probability and Related Topics},
  volume={4},
  pages={197-266},
  year={1978},
  publisher={Dekker New York}
}

@article{buchbinderSecretaryProblemsLinear2014,
  title={Secretary problems via linear programming},
  author={Buchbinder, Niv and Jain, Kamal and Singh, Mohit},
  journal={Mathematics of Operations Research},
  volume={39},
  number={1},
  pages={190--206},
  year={2014},
  publisher={INFORMS}
}

@inproceedings{chenProphetSecretaryMatching2025,
  title={Prophet secretary and matching: the significance of the largest item},
  author={Chen, Ziyun and Huang, Zhiyi and Li, Dongchen and Tang, Zhihao Gavin},
  booktitle={Proceedings of the 36th Annual ACM-SIAM Symposium on Discrete Algorithms},
  pages={1371--1401},
  year={2025}
}

@article{correaPostedPriceMechanisms2017,
  title={Posted price mechanisms and optimal threshold strategies for random arrivals},
  author={Correa, José and Foncea, Patricio and Hoeksma, Ruben and Oosterwijk, Tim and Vredeveld, Tjark},
  journal={Mathematics of Operations Research},
  volume={46},
  number={4},
  pages={1452--1478},
  year={2021},
  publisher={INFORMS}
}

@inproceedings{correaProphetInequalitiesIID2019,
  title={Prophet inequalities for {I.I.D.} random variables from an unknown distribution},
  author={Correa, José and D{\"u}tting, Paul and Fischer, Felix and Schewior, Kevin},
  booktitle={Proceedings of the 2019 ACM Conference on Economics and Computation},
  pages={3--17},
  year={2019}
}

@article{correaProphetSecretaryBlind2021,
  title={Prophet secretary through blind strategies},
  author={Correa, José and Saona, Raimundo and Ziliotto, Bruno},
  journal={Mathematical Programming},
  volume={190},
  number={1},
  pages={483--521},
  year={2021},
  publisher={Springer}
}

@article{correaRecentDevelopmentsProphet2019,
  title = {Recent Developments in Prophet Inequalities},
  author = {Correa, José and Foncea, Patricio and Hoeksma, Ruben and Oosterwijk, Tim and Vredeveld, Tjark},
  year = {2018},
  journal = {ACM SIGecom Exchanges},
  volume = {17},
  number = {1},
  pages = {61--70},
  urldate = {2023-11-07}
}

@article{correaSampledrivenOptimalStopping2021,
  title={Sample-driven optimal stopping: From the secretary problem to the {I}.{I}.{D}. prophet inequality},
  author={Correa, José and Cristi, Andrés and Epstein, Boris and Soto, José A.},
  journal={Mathematics of Operations Research},
  volume={49},
  number={1},
  pages={441--475},
  year={2024},
  publisher={INFORMS}
}

@article{correaTwosidedGameGoogol2020,
  title={The two-sided game of googol},
  author={Correa, José and Cristi, Andrés and Epstein, Boris and Soto, José A.},
  journal={Journal of Machine Learning Research},
  volume={23},
  number={113},
  pages={1--37},
  year={2022}
}

@inproceedings{correaUnknownIIDProphets2021,
  title = {Unknown {I}.{I}.{D}. Prophets: Better Bounds, Streaming Algorithms, and a New Impossibility (Extended Abstract)},
  booktitle = {12th Innovations in Theoretical Computer Science Conference},
  author = {Correa, José and Dütting, Paul and Fischer, Felix A. and Schewior, Kevin and Ziliotto, Bruno},
  year = {2021},
  pages = {86:1}
}

@inproceedings{cristiProphetInequalitiesRequire2023,
  title={Prophet inequalities require only a constant number of samples},
  author={Cristi, Andr{\'e}s and Ziliotto, Bruno},
  booktitle={Proceedings of the 56th Annual ACM Symposium on Theory of Computing},
  pages={491--502},
  year={2024}
}

@inproceedings{diakonikolasLearningOnlineAlgorithms2021,
  title = {Learning Online Algorithms with Distributional Advice},
  author = {Diakonikolas, Ilias and Kontonis, Vasilis and Tzamos, Christos and Vakilian, Ali and Zarifis, Nikos},
  booktitle={Proceedings of the 38th International Conference on Machine Learning},
  year = 2021,
  pages = {2687--2696}
}

@inproceedings{duttingPostedPricingProphet2019,
  title={Posted pricing and prophet inequalities with inaccurate priors},
  author={Dütting, Paul and Kesselheim, Thomas},
  booktitle={Proceedings of the 20th ACM Conference on Economics and Computation},
  pages={111--129},
  year={2019}
}

@article{dynkinOptimumChoiceInstant1963,
  title={The optimal choice of the stopping moment of a {Markov} process},
  author={Dynkin, Evgenii Borisovich},
  year = {1963},
  journal = {Soviet Mathematics},
  volume = {4},
  pages = {627--629}
}

@inproceedings{ehsaniProphetSecretaryCombinatorial2018,
  title={Prophet secretary for combinatorial auctions and matroids},
  author={Ehsani, Soheil and Hajiaghayi, MohammadTaghi and Kesselheim, Thomas and Singla, Sahil},
  booktitle={Proceedings of the 29th Annual ACM-SIAM Symposium on Discrete Algorithms},
  pages={700--714},
  year={2018}
}

@article{esfandiariProphetSecretary2015,
  title={Prophet secretary},
  author={Esfandiari, Hossein and Hajiaghayi, MohammadTaghi and Liaghat, Vahid and Monemizadeh, Morteza},
  journal={SIAM Journal on Discrete Mathematics},
  volume={31},
  number={3},
  pages={1685--1701},
  year={2017}
}

@inproceedings{esfandiariProphetsSecretariesMaximizing2020,
  title={Prophets, secretaries, and maximizing the probability of choosing the best},
  author={Esfandiari, Hossein and Hajiaghayi, MohammadTaghi and Lucier, Brendan and Mitzenmacher, Michael},
  booktitle={Proceedings of the 23rd International Conference on Artificial Intelligence and Statistics},
  pages={3717--3727},
  year={2020}
}

@article{ezraProphetInequalitySamples2024,
  title={Prophet Inequality from Samples: Is the More the Merrier?},
  author={Ezra, Tomer},
  journal={arXiv preprint arXiv:2409.00559},
  year={2024}
}

@inproceedings{feldmanImprovedCompetitiveRatios2011,
  title={Improved competitive ratios for submodular secretary problems},
  author={Feldman, Moran and Naor, Joseph and Schwartz, Roy},
  booktitle={Proceedings of the 14th International Workshop on Approximation Algorithms for Combinatorial Optimization},
  pages={218--229},
  year={2011}
}

@article{fergusonWhoSolvedSecretary1989,
  title={Who solved the secretary problem?},
  author={Ferguson, Thomas S},
  journal={Statistical science},
  volume={4},
  number={3},
  pages={282--289},
  year={1989},
  publisher={Institute of Mathematical Statistics}
}

@article{fujiiSecretaryProblemPredictions2024,
  title = {The Secretary Problem with Predictions},
  shorttitle = {The Secretary Problem with Predictions, {{FY}}, 24},
  author = {Fujii, Kaito and Yoshida, Yuichi},
  year = {2024},
  journal = {Mathematics of Operations Research},
  volume = {49},
  number = {2},
  pages = {1241--1262},
}

@misc{giambartolomeiProphetInequalitiesSeparating2023,
  title = {Prophet {{Inequalities}}: {{Separating Random Order}} from {{Order Selection}}},
  shorttitle = {Prophet {{Inequalities}}: {{Separating Random Order}} from {{Order Selection}}, {{GMS}}, 23},
  author = {Giambartolomei, Giordano and {Mallmann-Trenn}, Frederik and Saona, Raimundo},
  year = {2023},
  month = nov,
  number = {arXiv:2304.04024},
  eprint = {2304.04024},
  publisher = {arXiv},
  urldate = {2023-12-18},
  archiveprefix = {arXiv}
}

@article{gilbertRecognizingMaximumSequence,
  title={Recognizing the maximum of a sequence},
  author={Gilbert, John P and Mosteller, Frederick},
  journal={Journal of the American Statistical Association},
  volume={61},
  number={313},
  pages={35--73},
  year={1966}
}

@inproceedings{gravinOptimalProphetInequality2022,
  title = {Optimal Prophet Inequality with Less than One Sample},
  booktitle = {18th International Conference of Web and Internet Economics},
  author = {Gravin, Nick and Li, Hao and Tang, Zhihao Gavin},
  year = 2022,
  pages = {115--131}
}

@inproceedings{harbNewProphetInequalities2025,
  title={New Prophet Inequalities via Poissonization and Sharding},
  author={Harb, Elfarouk},
  booktitle={Proceedings of the 2025 Annual ACM-SIAM Symposium on Discrete Algorithms},
  pages={1222--1269},
  year={2025}
}

@article{hillComparisonsStopRule1982,
  title={Comparisons of stop rule and supremum expectations of iid random variables},
  author = {Hill, T. P. and Kertz, Robert P.},
  year = {1982},
  journal = {The Annals of Probability},
  volume = {10},
  number = {2},
  pages = {336--345},
  publisher = {Institute of Mathematical Statistics}
}

@article{kertzStopRuleSupremum1986,
  title = {Stop rule and supremum expectations of i.i.d. random variables: a complete comparison by conjugate duality},
  author = {Kertz, Robert P},
  year = {1986},
  journal = {Journal of Multivariate Analysis},
  volume = {19},
  number = {1},
  pages = {88--112}
}

@article{liProphetInequalityIID2023,
  title={{I}.{I}.{D}. prophet inequality with a single data point},
  author={Feng, Yilong and Li, Bo and Li, Haolong and Wu, Xiaowei and Wu, Yutong},
  journal={Artificial Intelligence},
  volume={341},
  pages={104296},
  year={2025},
  publisher={Elsevier}
}

@inproceedings{liuVariableDecompositionProphet2021,
  title = {Variable decomposition for prophet inequalities and optimal ordering},
  booktitle = {Proceedings of the 22nd ACM Conference on Economics and Computation},
  author = {Liu, Allen and Leme, Renato Paes and Pál, Martin and Schneider, Jon and Sivan, Balasubramanian},
  year = {2021},
  pages = {692}
}

@inproceedings{nutiSecretaryProblemDistributions2022,
  title={The secretary problem with distributions},
  author={Nuti, Pranav},
  booktitle={Proceedings of the 23rd International Conference on Integer Programming and Combinatorial Optimization},
  pages={429--439},
  year={2022}
}

@article{perezsalazarIIDProphetInequality2023,
  title={The {I}.{I}.{D}. Prophet Inequality with Limited Flexibility},
  author={Perez-Salazar, Sebastian and Singh, Mohit and Toriello, Alejandro},
  journal={Mathematics of Operations Research},
  year={2025},
  publisher={INFORMS}
}

@inproceedings{rubinsteinOptimalSinglechoiceProphet2020,
  title = {Optimal Single-Choice Prophet Inequalities from Samples},
  booktitle = {Proceedings of the 11th Innovations in Theoretical Computer Science Conference},
  author = {Rubinstein, Aviad and Wang, Jack Z. and Weinberg, S. Matthew},
  year = 2020,
  pages = {60:1--60:10}
}

@article{samuel1984comparison,
  title = {Comparison of Threshold Stop Rules and Maximum for Independent Nonnegative Random Variables},
  author = {{Samuel-Cahn}, Ester},
  year = {1984},
  journal = {The Annals of Probability},
  pages = {1213--1216}
}

@techreport{samuels1982,
  type = {Technical {{Report}}},
  title = {Exact Solutions for the Full Information Best Choice Problem},
  author = {Samuels, Stephen M.},
  year = {1982},
  number = {82-17},
  institution = {Department of Statistics, Purdue University}
}

@article{mitzenmacher2022algorithms,
  title={Algorithms with predictions},
  author={Mitzenmacher, Michael and Vassilvitskii, Sergei},
  journal={Communications of the ACM},
  volume={65},
  number={7},
  pages={33--35},
  year={2022},
  publisher={ACM New York, NY, USA}
}

@article{gnedin1994solution,
  title={A solution to the game of googol},
  author={Gnedin, Alexander V},
  journal={The Annals of Probability},
  pages={1588--1595},
  year={1994},
  publisher={JSTOR}
}

@inproceedings{azar2014prophet,
  title={Prophet inequalities with limited information},
  author={Azar, Pablo D and Kleinberg, Robert and Weinberg, S Matthew},
  booktitle={Proceedings of the 25th Annual ACM-SIAM Symposium on Discrete Algorithms},
  pages={1358--1377},
  year={2014}
}

@article{benomar2023advice,
  title={Advice querying under budget constraint for online algorithms},
  author={Benomar, Ziyad and Perchet, Vianney},
  journal={Advances in Neural Information Processing Systems},
  volume={36},
  pages={75026--75047},
  year={2023}
}

@article{braun2024secretary,
  title={The secretary problem with predicted additive gap},
  author={Braun, Alexander and Sarkar, Sherry},
  journal={Advances in Neural Information Processing Systems},
  volume={37},
  pages={16321--16341},
  year={2024}
}

@article{canonne2025little,
  title={With a Little Help From My Friends: Exploiting Probability Distribution Advice in Algorithm Design},
  author={Canonne, Cl{\'e}ment L and Chen, Kenny and Mestre, Juli{\'a}n},
  journal={arXiv preprint arXiv:2505.04949},
  year={2025}
}

@article{lykouris2021competitive,
  title={Competitive caching with machine learned advice},
  author={Lykouris, Thodoris and Vassilvitskii, Sergei},
  journal={Journal of the ACM},
  volume={68},
  number={4},
  pages={1--25},
  year={2021},
  publisher={ACM New York, NY}
}

@article{purohit2018improving,
  title={Improving Online Algorithms via ML Predictions},
  author={Purohit, Manish and Svitkina, Zoya and Kumar, Ravi},
  journal={Advances in Neural Information Processing Systems},
  volume={31},
  year={2018}
}

@article{kehne2025prophet,
  title={Prophet and Secretary at the Same Time},
  author={Kehne, Gregory and Kesselheim, Thomas},
  journal={arXiv preprint arXiv:2511.09531},
  year={2025}
}

@article{balkanski2024fair,
  title={Fair secretaries with unfair predictions},
  author={Balkanski, Eric and Ma, Will and Maggiori, Andreas},
  journal={Advances in Neural Information Processing Systems},
  volume={37},
  pages={3682--3716},
  year={2024}
}

@inproceedings{esfandiari2015online,
  title={Online allocation with traffic spikes: Mixing adversarial and stochastic models},
  author={Esfandiari, Hossein and Korula, Nitish and Mirrokni, Vahab},
  booktitle={Proceedings of the 16th ACM Conference on Economics and Computation},
  pages={169--186},
  year={2015}
}

@article{angelopoulos2024online,
  title={Online computation with untrusted advice},
  author={Angelopoulos, Spyros and D{\"u}rr, Christoph and Jin, Shendan and Kamali, Shahin and Renault, Marc},
  journal={Journal of Computer and System Sciences},
  volume={144},
  pages={103545},
  year={2024},
  publisher={Elsevier}
}

\appendix

\section{Derivation of Optimal Robust Threshold Functions}
\label{app:unique-inheritance-optimal-threshold}

In this section, we prove~\Cref{lem:unique-inheritance-optimal-threshold}. First, we derive the indifference (Bellman) equations for the bi-criteria algorithm under the \MaxExp and \MaxProb objectives. Then, we show that, for each objective, these equations uniquely determine the optimal threshold functions~$\maxexpoptthres$ and $\maxproboptthres$ when the predicted prior coincides with the true prior~$F$. We first treat the \MaxExp objective and then turn to \MaxProb.

\subsection{MaxExp: Optimal Threshold}
\label{app:maxexp-threshold}

\begin{lemma}[Optimal threshold function for \MaxExp]
    \label{lem:indifference-existence-uniqueness-MaxExp}
    For any $n\in \mathbb{N}$ and any prior $F$, the optimal threshold function
    $\maxexpoptthres$ that maximizes the consistency of $\algowiththres(F,\threshold)$ among all $\threshold$'s under \MaxExp is uniquely determined at each time $t$ by
    $\maxexpoptthres(t) = F(x^*(t))$, where $x^*(t)$ is the unique solution of
    \begin{equation}
        \label{eq:max-exp-back-induction-app}
        \phi(t,x) ~=~ 0,
    \end{equation}
    with
    \[
        \phi(t,x)
        ~\coloneqq~
        \E_{y\sim F}[y\mid y\geq x]
        \left( 1-\left(\frac{F(x)}{1-t+tF(x)}\right)^{n-1}\right) - x.
    \]
    Moreover, $\maxexpoptthres(t)$ is non-increasing in $t$ on $[0,1]$.
\end{lemma}

\begin{proof}
    Fix $t\in[0,1]$ and suppose the thresholds $\threshold(s)$ have already been fixed
    for all $s\in(t,1]$.
    Let $x$ denote the value that arrives at time $t$, and assume $x$ is a record (else the algorithm cannot accept at $t$).
    By definition,
    the algorithm $\algowiththres(F,\maxexpoptthres)$ accepts this value $x$ at time $t$ if and only if $F(x) > \maxexpoptthres(t)$.
    By optimality, the threshold chosen at time $t$ must make the expected rewards
    from accepting and rejecting this value $x$ equal.
    
    \medskip\noindent\textbf{Accepting.}
    Accepting yields a reward $x$.
    
    \medskip\noindent\textbf{Rejecting.}
    If the algorithm rejects $x$, then by the non-increasing threshold property, it must accept the first later value $y \ge x$.
    Thus, the continuation reward is
    \[
       \E_{y\sim F}[y \mid y \ge x]\;\cdot\;
       \Pr\!\big[\exists\, y \ge x \text{ after } t \mid x \text{ is a record at } t\big].
    \]
    
    To compute the existence probability, note that all other $n-1$ samples lie
    below $x$ with probability $F(x)^{n-1}$, and each sample is either below $x$ or arrives after $t$ with probability $1 - t + tF(x)$.
    Conditioning on the latter event gives
    \[
        \Pr\!\big[\exists\, y \ge x \text{ after } t \mid x \text{ is a record at } t\big]
        = 1 - \frac{F(x)^{n-1}}{(1 - t + tF(x))^{\,n-1}}.
    \]
    This conditioning introduces the factor $(1-t+tF(x))^{n-1}$, because each of the other $n-1$ values must either arrive after time $t$ or fall below $x$.
    Equating the accept and reject rewards yields exactly the equation $\phi(t,x)=0$
    in \Cref{eq:max-exp-back-induction-app}.
    
    By Lemma~\ref{lem:opt-threshold-unique-max-expect}, for any continuous prior $F$ with full support on $[0,\infty)$, any $n\in \mathbb{N}$, and any $t\in[0,1]$, the equation $\phi(t,x)=0$ has a unique solution $x^*(t)\in[0,\infty)$.
    Thus, defining $\maxexpoptthres(t) \coloneqq F(x^*(t))$ yields a unique optimal threshold at each time $t$.
    
    Finally, by Lemma~\ref{lem:opt-threshold-monotone-max-exp}, $x^*(t)$ is non-increasing in $t$ on $[0,1]$.
    Since $F$ is non-decreasing, the same holds for $\maxexpoptthres(t)=F(x^*(t))$.
    
    This completes the proof.
\end{proof}

We now establish existence, uniqueness, and monotonicity of the solution $x^*(t)$ to \Cref{eq:max-exp-back-induction-app}.
For the proofs it will be convenient to introduce some auxiliary notation.
For $u\in[0,1)$ and $t\in[0,1)$, define
\[
    g(u,t) \coloneqq \frac{1}{1-u}
    \left(1-\left(\frac{u}{1-t+tu}\right)^{n-1}\right),
\]
and for any $(t,x)$ with $t\in[0,1)$ and $x\ge 0$ define
\[
    h(t,x)
    \coloneqq
    \frac{x}{g(F(x),t)} - \int_x^\infty y\,\dd F(y).
\]
Using
\[
    \E_{y\sim F}[y\mid y\ge x]
    =
    \frac{\int_x^\infty y\,\dd F(y)}{1-F(x)},
\]
one checks that for each fixed $t\in[0,1)$ and $x\ge 0$,
\begin{equation}
    \label{eq:Gt-Ht-equivalence}
    \phi(t,x)=0
    \quad\Longleftrightarrow\quad
    h(t,x)=0.
\end{equation}
We first prove the existence and uniqueness of $x^*(t)$ for each $t$, and then show that $x^*(t)$ is non-increasing in $t$.

\begin{lemma}
    \label{lem:opt-threshold-unique-max-expect}
    For any continuous prior $F$ with full support $[0,\infty)$, any $n\in\mathbb{N}$,
    and any $t\in[0,1]$, \Cref{eq:max-exp-back-induction-app} has a unique solution
    $x^*(t)\in[0,\infty)$.
\end{lemma}

\begin{proof}
    When $t=1$, the factor
    \[
        1-\left(\frac{F(x)}{1-t+tF(x)}\right)^{n-1}
    \]
    equals $0$ for all $x\geq 0$, so $\phi(1,x) = -x$ has a unique zero $x^*(1) = 0$.
    We focus on the case $0\leq t<1$ below.
    
    For fixed $t\in[0,1)$, the equivalence \eqref{eq:Gt-Ht-equivalence} implies that it suffices to show that $h(t,x)$ is strictly increasing in $x$ on $[0,\infty)$ and crosses zero exactly once.
    
    First, $g(F(0),t)=g(0,t)=1$, and by L'H\^opital's rule
    \[
        \lim_{x\to\infty} g(F(x),t) = \lim_{u\to 1} g(u,t)=(n-1)(1-t).
    \]
    Hence
    \begin{align*}
        h(t,0)&=-\int_{0}^{\infty} y \,\dd F(y) \leq 0,\\
        \lim_{x\to \infty} h(t,x)&=\lim_{x\to\infty} \frac{x}{(n-1)(1-t)} = \infty,
    \end{align*}
    so, by the Intermediate Value Theorem, $h(t,\cdot)$ has at least one root in $[0,\infty)$.
    
    We now show that $h(t,x)$ has at most one root in $[0,\infty)$.
    It suffices to prove that $\partial_x h(t,x)>0$ for any $x \geq 0$.
    Differentiating in $x$, we obtain
    \[
        \partial_x h(t,x)
        =
        \frac{1}{g(F(x),t)}
        +
        \frac{x\,F'(x)\big(g(F(x),t)^2-\partial_u g(F(x),t)\big)}{g(F(x),t)^2},
    \]
    where $F'(x)$ denotes the density of $F$ (which exists almost everywhere).
    Since $g(u,t)>0$ for any $u\in[0,1)$, it remains to show that for any $u\in[0,1)$,
    \[
        g(u,t)^2 - \partial_u g(u,t) \geq 0.
    \]
    
    Differentiating in $u$,
    \[
        g(u,t)^2-\partial_u g(u,t)
        = \frac{u^{n-2} \left[ u^n - u (1-t+tu)^{n-1}
        + (n-1)(1-t)(1-u) (1-t+tu)^{n-2} \right]}{(1-u)^2 (1-t+tu)^{2n-2}}.
    \]
    Let $z(u,t)$ denote the bracketed term in the numerator.
    Then $z(u,1)=0$ for any $u\in[0,1)$, and
    \[
        \partial_t z(u,t)
        =
        -(n-1)(1-u)^2 (1-t+tu)^{n-3} \left[ (1-t+tu) + (1-t)(n-2) \right] \leq 0.
    \]
    Thus $z(u,t) \geq z(u,1)=0$ for any $u,t\in[0,1)$, which implies
    $g(u,t)^2 - \partial_u g(u,t) \ge 0$ and hence $\partial_x h(t,x) > 0$
    for all $x \ge 0$.
    
    Therefore, for each fixed $t\in[0,1)$, the function $h(t,\cdot)$ is strictly increasing on $[0,\infty)$ and crosses zero exactly once.
    By \eqref{eq:Gt-Ht-equivalence}, the equation $h(t,x)=0$ has a unique solution $x^*(t)\in[0,\infty)$ for each fixed $t\in[0,1)$; and together with the case $t=1$, this proves the claim and the lemma.
\end{proof}

\begin{lemma}
    \label{lem:opt-threshold-monotone-max-exp}
    For any continuous prior $F$ with full support $[0,\infty)$ and any $n\in \mathbb{N}$,
    the unique solution $x^*(t)$ of \Cref{eq:max-exp-back-induction-app} is non-increasing in $t$ on $[0, 1]$.
\end{lemma}

\begin{proof}
   With $g(u,t)$ and $h(t,x)$ defined above, Lemma~\ref{lem:opt-threshold-unique-max-expect} implies that for each fixed $t\in[0,1)$ the function $h(t,\cdot)$ is strictly increasing on $[0,\infty)$ and has a unique zero at $x=x^*(t)$.

   Fix $x\ge 0$ and let $u=F(x)$.
   For $u\in[0,1)$ and $t\in[0,1)$,
   \[
      \partial_t g(u,t)
      =
      -\frac{(n-1)u^{n-1}}{(1-t+tu)^n} \le 0,
   \]
   so $g(F(x),t)$ is non-increasing in $t$, $1/g(F(x),t)$ is non-decreasing, and hence
   \[
   \partial_t h(t,x)
      =
      x\,\partial_t\left[\frac{1}{g(F(x),t)}\right]\ge 0.
   \]

   Thus, for any $0\le t_1<t_2\le 1$ and any $x \ge 0$,
   \[
      h(t_1,x) \le h(t_2,x).
   \]

   Fix $0\le t_1<t_2\le 1$.
   Since $x^*(t_2)$ is the unique zero of $h(t_2,\cdot)$, we have $h(t_2,x^*(t_2))=0$, and the inequality above gives $h(t_1,x^*(t_2))\le 0$.
   Because $h(t_1,\cdot)$ is strictly increasing on $[0,\infty)$ and has a unique zero at $x^*(t_1)$, this implies $x^*(t_1)\ge x^*(t_2)$.

   Hence $x^*(t)$ is non-increasing in $t$ on $[0,1]$.
\end{proof}

\subsection{MaxProb: Optimal Threshold}
\label{app:maxprob-threshold}

\begin{lemma}[Optimal threshold function for \MaxProb]
   \label{lem:indifference-existence-uniqueness-MaxProb}
   For any $n\in \mathbb{N}$ and any prior $F$, the optimal threshold function $\maxproboptthres$ that maximizes the consistency of $\algowiththres(F,\maxproboptthres)$ under \MaxProb is uniquely determined at each time $t$ by the unique solution $\maxproboptthres(t)\in[0,1]$ of
   \begin{equation}
      \label{eq:max-prob-back-induction-integral}
      \psi(t,\quantile) ~=~ 0,
   \end{equation}
   where
   \[
      \label{eq:maxprob-opt-threshold-function}
      \psi(t,\quantile)
      ~\coloneqq~
      \int_{t}^1 \frac{(1-s+s\quantile)^{n-1}-\quantile^{n-1}}{1-s}\,\dd s - \quantile^{n-1}.
   \]
   Moreover, $\maxproboptthres(t)$ is non-increasing in $t$ on $[0,1]$.
\end{lemma}

Note that the form of Equation~\eqref{eq:max-prob-back-induction-integral} coincides with the equation defining the threshold function used by~\Cref{example:gilbert-mosteller}.

\begin{proof}
    Fix $t\in[0,1]$ and suppose the thresholds $\threshold(s)$ have already been fixed for all $s\in(t,1]$.
    Let $x$ denote the value that arrives at time $t$, and assume $x$ is a record (otherwise the algorithm cannot accept at $t$).
    By definition, the algorithm $\algowiththres(F,\threshold)$ accepts this value $x$ at time $t$ if and only if $F(x)>\threshold(t)$.
    By optimality, the threshold chosen at time $t$ must make the success probabilities from accepting and rejecting this value $x$ equal.
    
    \medskip\noindent\textbf{Accepting.}
    If the algorithm accepts a record $x$ at time $t$, the success event is that $x$ is the largest among all $n$ values.
    Conditional on $x$ being a record at time $t$, this occurs with probability
    \[
       \frac{\Pr[x \text{ is the largest among all values}]}{\Pr[x \text{ is a record at time }t]}
       ~=~ \frac{F(x)^{n-1}}{(1-t+tF(x))^{n-1}}.
    \]
    
    \medskip\noindent\textbf{Rejecting.}
    If instead the algorithm rejects such an $x$, then by the non-increasing threshold assumption it will accept the next value larger than $x$, which is necessarily the next record that exceeds the threshold.
    Given that there are in total $k$ values larger than $x$, the probability that the algorithm eventually picks the largest among all values is $1/k$, since the arrival order is uniform at random and the algorithm accepts the first exceedance over $x$.
    The probability that there are exactly $k$ values larger than $x$ and arriving after time $t$, conditional on $x$ being a record at time $t$, is
    \[
       \frac{\binom{n-1}{k} \big((1-t)(1-F(x))\big)^k F(x)^{n-1-k}}{(1-t+tF(x))^{n-1}}.
    \]
    
    Equating the success probability from accepting and rejecting $x$ yields the
    backward-induction equation
    \begin{equation}
        \label{eq:max-prob-back-induction-app}
        \sum_{k=1}^{n-1} \binom{n-1}{k} \frac{1}{k} \big((1-t)(1-F(x))\big)^k F(x)^{n-1-k}
        ~=~ F(x)^{n-1}.
    \end{equation}
    
    For the quantile threshold $\quantile = F(x)$, this becomes
    \[
        \sum_{k=1}^{n-1} \binom{n-1}{k} \frac{1}{k}
        \big((1-t)(1-\quantile)\big)^k \quantile^{n-1-k}
        ~=~ \quantile^{n-1}.
    \]
    
    Using the identity $\frac{1}{k} = \int_0^1 y^{k-1}\,\dd y$, the left-hand side can be written as
    \[
    \begin{aligned}
        \sum_{k=1}^{n-1} \binom{n-1}{k} \frac{1}{k}\big((1-t)(1-\quantile)\big)^k \quantile^{n-1-k}
        &= \sum_{k=1}^{n-1} \binom{n-1}{k} \quantile^{n-1-k}
           \int_{0}^{(1-t)(1-\quantile)} y^{k-1} \,\dd y\\
        &= \int_0^{(1-t)(1-\quantile)} \frac{1}{y}
           \sum_{k=1}^{n-1} \binom{n-1}{k} \quantile^{n-1-k}y^{k}\,\dd y\\
        &= \int_0^{(1-t)(1-\quantile)} \frac{(y+\quantile)^{n-1}-\quantile^{n-1}}{y} \,\dd y\\
        &= \int_{t}^1 \frac{(1-s+s\quantile)^{n-1}-\quantile^{n-1}}{1-s}\,\dd s,
    \end{aligned}
    \]
    where the last equality follows from the change of variables $y=(1-s)(1-\quantile)$.
    Under this substitution, $y=0$ corresponds to $s=1$ and $y=(1-t)(1-\quantile)$ corresponds to $s=t$, which reverses the limits.
    Thus, \Cref{eq:max-prob-back-induction-app} is equivalent to
    \[
        \int_{t}^1 \frac{(1-s+s\quantile)^{n-1}-\quantile^{n-1}}{1-s}\,\dd s
        ~=~ \quantile^{n-1},
    \]
    which is precisely the equation $\psi(t,\quantile)=0$ in
    \Cref{eq:max-prob-back-induction-integral}.
    
    We now prove existence and uniqueness of the solution of $\psi(t,\quantile)=0$
    for each fixed $t$, and then establish the monotonicity of $\maxproboptthres(t)$.
    
    For each fixed $t\in[0,1)$, the function $\psi(t,\quantile)$ is continuous on
    $\quantile \in [0,1]$.
    At $\quantile=0$, the integrand equals $(1-s)^{n-2}$, hence
    \[
    \psi(t,0)
    =
    \int_t^1 (1-s)^{n-2}\,\dd s
    > 0.
    \]
    At $\quantile=1$, the integrand is identically zero, so
    \[
    \psi(t,1)
    =
    0 - 1^{n-1}
    = -1 < 0.
    \]
    Thus, for each $t\in[0,1)$ the equation $\psi(t,\quantile)=0$ has at
    least one root $\quantile \in (0,1)$.
    
    For uniqueness, observe first that when $t = 1$, the integrand in $\psi(1,\quantile)$
    is zero for all $\quantile\in[0,1]$, so $\psi(1,\quantile)=-\quantile^{n-1}$ and
    $\maxproboptthres(1)=0$ is the unique root.
    For $t < 1$, define
    \[
        \tilde{\psi}(t,\quantile)
        \coloneqq \quantile^{-(n-1)} \psi(t,\quantile)
        =
        \int_{t}^1 \frac{\left(\frac{1-s+s\quantile}{\quantile}\right)^{n-1}-1}{1-s}\,\dd s - 1,
        \qquad \quantile\in(0,1).
    \]
    Differentiating with respect to $\quantile$ gives
    \[
        \partial_\quantile \tilde{\psi}(t,\quantile)
        =
        -\frac{n-1}{\quantile^n} \int_t^1 (1-s+s\quantile)^{n-2} \,\dd s < 0,
    \]
    so for each fixed $t<1$ the function $\tilde{\psi}(t,\quantile)$ is strictly
    decreasing in $\quantile$ on $(0,1)$.
    Thus, for each $t\in[0,1)$ the equation $\tilde{\psi}(t,\quantile)=0$ has a unique
    solution $\maxproboptthres(t)\in(0,1)$.
    Because $\tilde{\psi}(t,\quantile)=0$ is equivalent to $\psi(t,\quantile)=0$ for $\quantile>0$,
    the solution of $\psi(t,\quantile)=0$ is also unique.
    
    Finally, it remains to show that $\maxproboptthres(t)$ is non-increasing in $t$.
    Fix $0\le t_1<t_2<1$ and suppose, for contradiction, that $\maxproboptthres(t_2)>\maxproboptthres(t_1)$.
    Since $\maxproboptthres(t_i)$ is the unique zero of $\tilde{\psi}(t_i,\quantile)$, we have
    \[
        \tilde{\psi}(t_i,\maxproboptthres(t_i))=0,\qquad i=1,2.
    \]
    
    For fixed $t_1$, the function $\quantile\mapsto\tilde{\psi}(t_1,\quantile)$ is strictly
    decreasing on $(0,1)$ and has a unique zero at $\maxproboptthres(t_1)$.
    The inequality $\maxproboptthres(t_2)>\maxproboptthres(t_1)$ therefore implies
    \begin{equation}
        \label{eq:maxprob-ineq-theta}
        \tilde{\psi}(t_1,\maxproboptthres(t_2))<\tilde{\psi}(t_1,\maxproboptthres(t_1))=0.
    \end{equation}
    
    On the other hand, for fixed $\quantile=\maxproboptthres(t_2)$ the function
    $\tilde{\psi}(t,\quantile)$ is strictly decreasing on $[0,1)$.
    Since $t_1<t_2$, this yields
    \begin{equation}
        \label{eq:maxprob-ineq-t}
        \tilde{\psi}(t_1,\maxproboptthres(t_2))
        >
        \tilde{\psi}(t_2,\maxproboptthres(t_2))
        =0.
    \end{equation}
    The inequalities \eqref{eq:maxprob-ineq-theta} and \eqref{eq:maxprob-ineq-t}
    contradict each other.
    Hence the assumption $\maxproboptthres(t_2)>\maxproboptthres(t_1)$ is impossible, and we must have
    $\maxproboptthres(t_2)\le\maxproboptthres(t_1)$ for all $0\le t_1<t_2<1$.
    Thus $\maxproboptthres(t)$ is non-increasing on $[0,1)$, and together with
    $\maxproboptthres(1)=0$ this completes the proof.
\end{proof}

\section{Numerical Solution to the Differential Equation in \Cref{thm:max-exp-consistency-robustness}}
\label{app:max-expect-numerical-solution}

Recall the differential equation in~\Cref{thm:max-exp-consistency-robustness}: 
\begin{align}
    \threshold(\lambda_1) & ~=~ 1 ~, \label{eq:max-expect-competitive-boundary-conditions}\\
    \int_z^1 \int_0^t \frac{1}{t} \threshold \big( \max \{s, z \} \big)^t \dd s \dd t 
    & ~\ge~ \alpha \cdot \threshold(z)
    ~, \qquad 
    \forall  z\in [\lambda_1,\lambda_2] 
    ~. \label{eq:max-expect-competitive-interval-conditions}
\end{align}
We assume throughout that $\threshold$ is left‑continuous, non‑increasing, and $\threshold(s)\equiv 0$ for $s\in(\lambda_2,1]$.

Define the left-hand side of \Cref{eq:max-expect-competitive-interval-conditions} as $L(z)$ and simplify: 

\begin{align}
   L(z)~:=&~\int_z^1 \int_0^t \frac{1}{t} 
   \threshold \big(\max\{s, z \} \big)^t \dd s \dd t \label{eq:max-expect-Lz-simplify} \\
   ~=&~ \int_z^1 \bigg( \frac{z\,\threshold(z)^t}{t}  + \int_z^t \frac{\threshold(s)^t}{t} \dd s \bigg) \dd t \notag \\
   ~=&~ z \int_{z}^{1}\frac{\threshold(z)^t}{t} \dd t +\int_{z}^{\lambda_2} \int_s^{1} \frac{\threshold(s)^t}{t} \dd t \dd s\notag ~.
\end{align}
Here, the last equality is obtained by exchanging the order of integration and using $\threshold\equiv 0$ on $(\lambda_2,1]$.

We approximate $\threshold$ by a left‑continuous non‑increasing step function computed backward from $\lambda_2$ to $\lambda_1$ so that \Cref{eq:max-expect-competitive-interval-conditions} holds for all $z\in[\lambda_1,\lambda_2]$.
\begin{definition}[Step function approximation]
    \label{def:max-exp-step-function-approximation}
    Given an integer $m>0$, define a uniform grid $\{z_i\}_{i=1}^{m+1}$ by $z_i=\lambda_1+\tfrac{i-1}{m}(\lambda_2-\lambda_1)$ for $i=1,\dots,m+1$ (so $z_{m+1}=\lambda_2$). We say that $\threshold$ is a step function with step values $\threshold_1,\dots,\threshold_m$ if $\threshold(z)=\threshold_i$ for all $z\in(z_i,z_{i+1}]$, and $\threshold(\lambda_1)=\threshold_1$.
\end{definition}

\begin{lemma}[Monotonicity of $L(z)$]
    \label{lem:max-exp-L-monotone}
    $L(z)$ is decreasing in $z$ on $[\lambda_1,\lambda_2]$.
\end{lemma}

\begin{proof}
    As $z$ increases, the integration domain in \Cref{eq:max-expect-Lz-simplify} shrinks with a non-negative integrand.
    The integrand is non-increasing since $\threshold$ is non-increasing.
\end{proof}

\begin{lemma}[Endpoint reduction]
    \label{lem:max-exp-endpoint-reduction}
    A step function $\threshold$ with $m$ steps satisfies \Cref{eq:max-expect-competitive-interval-conditions} for all $z\in[\lambda_1,\lambda_2]$ if it satisfies \Cref{eq:max-expect-competitive-interval-conditions} at the endpoints of each interval, i.e., for all $i \in [m]$, 
    \[
        \label{eq:max-expect-competitive-endpoint-conditions}
        L(z_{i+1}) ~\ge~ \alpha\, \threshold_i~.
    \]
\end{lemma}
\begin{proof}
By Lemma~\ref{lem:max-exp-L-monotone}, $L(z)$ is decreasing in $z$.
Thus, for any $z\in(z_i,z_{i+1}]$, it holds $L(z)\ge L(z_{i+1})\ge \alpha\,\threshold_i=\alpha\,\threshold(z)$.
\end{proof}

Next, we describe how to compute the step values $\threshold_1,\dots,\threshold_m$ given $\alpha$ and $m$ so that \Cref{eq:max-expect-competitive-endpoint-conditions} holds.

\paragraph{Initialization at $z=\lambda_2$.}
For $i=m$, \Cref{eq:max-expect-competitive-endpoint-conditions} is
\[
    \lambda_2 \int_{\lambda_2}^1 \frac{\threshold(\lambda_2)^t}{t} \,\dd t ~\ge~ \alpha\,\threshold(\lambda_2)~.
\]
Letting $\threshold_m\coloneqq\threshold(\lambda_2)\ne 0$, this is 
\[
    \lambda_2 \int_{\lambda_2}^1 \frac{\threshold_m^{-(1-t)}}{t}\,\dd t ~\ge~ \alpha~,
\]
where the left-hand side decreases from $+\infty$ to $0$ as $\threshold_m$ increases from $0$ to $+\infty$. Thus, there exists a $\threshold_m^* > 0$ such that the equality holds.\footnote{During experiments, we use binary search to locate $\threshold_m^*$ with $100$ iterations so that the error can be ignored.}

\paragraph{Backward recursion.}
For $i = m - 1, m - 2, \ldots,1$, suppose we have computed $\threshold_{i+1}^*, \threshold_{i+2}^*, \ldots, \threshold_{m}^*$. Using \Cref{eq:max-expect-Lz-simplify}, we have
\[
    \begin{aligned}
        L(z_{i+1})
        ~=~& z_{i+1} \int_{z_{i+1}}^1 \frac{ \threshold_i^t}{t} \,\dd t
        +\sum_{j=i+1}^m \int_{z_{j}}^{z_{j+1}} \int_s^1 \frac{{\threshold^*_j}^t}{t} \,\dd t \,\dd s\\
        ~=~& z_{i+1} \int_{z_{i+1}}^1 \frac{ \threshold_i^t}{t} \,\dd t + \sum_{j=i+1}^m \left[ \int_{z_{j}}^{z_{j+1}} \frac{t- z_j}{t}\,{\threshold_j^*}^t  \,\dd t   + (z_{j+1}-z_j) \int_{z_{j+1}}^1 \frac{{\threshold^*_j}^t}{t}  \,\dd t \right]
        ~.
    \end{aligned}
\]
Thus, \Cref{eq:max-expect-competitive-endpoint-conditions} is simplified to
\begin{align}
    \label{eq:max-exp-numerical-Lz}
    z_{i+1} \int_{z_{i+1}}^1 \frac{ \threshold_i^t}{t} \,\dd t + \sum_{j=i+1}^m \left[ \int_{z_{j}}^{z_{j+1}} \frac{t- z_j}{t}\,{\threshold_j^*}^t  \,\dd t   + (z_{j+1}-z_j) \int_{z_{j+1}}^1 \frac{{\threshold^*_j}^t}{t}  \,\dd t \right]
    ~\ge~ \alpha\, \threshold_i~.
\end{align}
Equivalently,
\begin{equation}
\label{eq:max-expect-ratio-Lz}
    z_{i+1} \int_{z_{i+1}}^1 \frac{ \threshold_i^{-(1-t)}}{t} \,\dd t
    + \threshold_i^{-1}\, \sum_{j=i+1}^m \left[ \int_{z_{j}}^{z_{j+1}} \frac{t- z_j}{t}\,{\threshold^*_j}^t  \,\dd t   + (z_{j+1}-z_j) \int_{z_{j+1}}^1 \frac{{\threshold^*_j}^t}{t}  \,\dd t \right]
    ~\ge~ \alpha~,
\end{equation}
whose left-hand side decreases from $\infty$ to $0$ as $\threshold_i$ increases from $0$ to $\infty$. Again, there exists a $\threshold_i^* > 0$ such that the equality holds.

We select the largest feasible $\threshold_i=\threshold_i^*$ by binary search such that equality holds.

\begin{lemma}[Monotonicity of the step values]
    \label{lem:max-exp-thetai-monotone}
    The sequence $(\threshold_i^*)_{i=1}^m$ produced by the recursion is non-increasing in $i$.
    \end{lemma}
\begin{proof}
   By definition, $\alpha \threshold_i^*=L(z_{i+1})\geq L(z_{i+2})=\alpha \threshold_{i+1}^*$, where the inequality holds by \Cref{lem:max-exp-L-monotone}.
\end{proof}

\paragraph{Enforcing $\threshold(\lambda_1)=1$.}
We only need to find an $\alpha$ with the corresponding sequence $(\threshold_i^*)_{i=1}^m$ satisfying~\Cref{eq:max-expect-competitive-interval-conditions} and $\threshold_1^*\geq 1$.
The following lemma shows this suffices.
\begin{lemma}
    \label{lem:max-exp-theta1-ge-1}
    If a step function $\threshold$ with step values $\threshold_1, \threshold_2, \dots, \threshold_m$ satisfies~\Cref{eq:max-expect-competitive-interval-conditions}, then the step function with step values $(\min\{\threshold_i,1\})_{i=1}^m$ also satisfies~\Cref{eq:max-expect-competitive-interval-conditions}.
\end{lemma}
\begin{proof}
    By merging intervals where $\threshold_i$ exceeds $1$, and by Lemma~\ref{lem:max-exp-endpoint-reduction}, we only need to check that \Cref{eq:max-expect-ratio-Lz} holds for $\min\{\threshold_i,1\}$, where $i$ is the last index such that $\threshold_i>1$. This follows immediately from the non-increasing property of the left-hand side.
\end{proof}

A direct corollary implies the desired result.
\begin{corollary}
    \label{cor:max-exp-theta1-ge-1}
    If there exists an $\alpha$ such that the corresponding recursively constructed sequence $(\threshold_i^*)_{i=1}^m$ satisfies $\threshold_1^*\geq 1$, then the step function $\threshold^*$ with step values $(\min\{\threshold_i^*,1\})_{i=1}^m$ is a left-continuous non-increasing function satisfying \Cref{eq:max-expect-competitive-boundary-conditions} and \Cref{eq:max-expect-competitive-interval-conditions}.
\end{corollary}

In our implementation, we fix several pairs of values $(\alpha,\beta)$ and numerically verify that for each pair, the constructed sequence $(\threshold_i^*)_{i=1}^m$ satisfies $\threshold_1^*\geq 1$.
Then~\Cref{cor:max-exp-theta1-ge-1} implies that the algorithm $\mathcal{A}(F,(\threshold^*)^{\nicefrac{1}{n}})$ is $\alpha$-consistent and $\beta$-robust.
The results are plotted in~\Cref{fig:MaxExp}.

Intuitively, as $m\to\infty$, the step function approximation tightens.
In experiments, we set $m=300$.
For details, see the implementation of the backward recursion as well as the solved $\threshold_i^*$'s
in \href{https://github.com/billyldc/optimal-stopping-with-a-predicted-prior}{our code repository}.\footnote{\url{https://github.com/billyldc/optimal-stopping-with-a-predicted-prior}}

\Cref{fig:theta-illustration} illustrates the case where $\beta=0.01$ and $m=300$, which yields $\alpha\approx 0.6908$ with $\threshold^*_1\approx 1.0165$.
This implies that $\mathcal{A}(F,(\threshold^*)^{\nicefrac{1}{n}})$ with $\threshold^*$ being the step function with step values $(\min\{\threshold_i^*,1\})_{i=1}^{300}$ is an $\alpha\approx 0.6908$-consistent and $\beta=0.01$-robust algorithm.

\begin{figure}[H]
    \centering
    \includegraphics[width=0.6\linewidth]{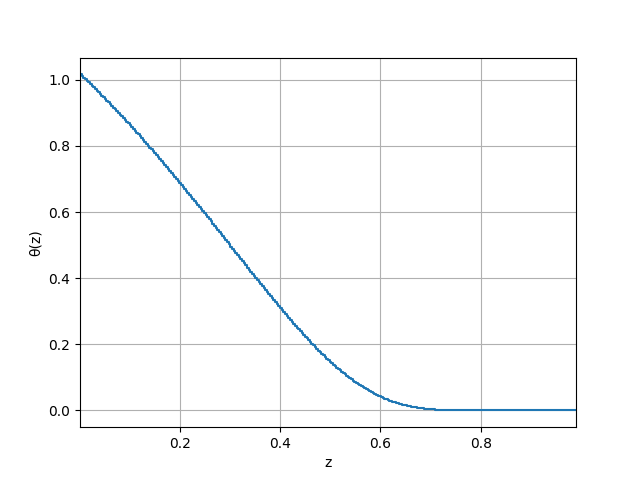}
    \caption{Example solution of $\threshold$ when $\beta=0.01$, $m=300$, and $\alpha\approx 0.6908$.}
    \label{fig:theta-illustration}
\end{figure}

\section{Missing Proofs for \MaxProb Consistency}
\label{app:algorithm-analysis-max-prob-consistency}

\begin{lemma}
    \label{lem:max-prob-asymptotic}
    For any fixed $s\in [0,1)$, we have
    \begin{equation*}
        \maxproboptthres(s) ~=~ \frac{1}{1+\frac{\constint}{(n-1)(1-s)}} + o\Big(\frac{1}{n}\Big)
            ~,
    \end{equation*}
    where $\constint\approx 0.80435$ satisfies
    \[
        \sum_{k\ge1} \frac{\constint^k}{k!k}=1
        ~.
    \]
\end{lemma}

\begin{proof}

    Recall \Cref{eqn:max-prob-foc-alt}:
    \begin{equation}
        \label{eqn:max-prob-foc-alt}
        \int_s^1 \frac{\big(\frac{1-t}{\maxproboptthres(s)} + t \big)^{n-1} - 1}{1-t} \dd t = 1~.
    \end{equation}
    
    Fix any $s$, substitute $x_n=\frac{1}{\maxproboptthres(s)} - 1$, the left-hand side becomes
    \[
        \int_s^1 \frac{\left(1 + x_n(1-t)\right)^{n-1} - 1}{1-t} \dd t~.
    \]
    
    Expanding the left-hand side with the binomial theorem, we have
    \[
        (n-1)(1-s)x_n + \binom{n-1}{2} \frac{(1-s)^2}{2} x_n^2  + \cdots + \binom{n-1}{k}  \frac{(1-s)^k}{k}x_n^k + \cdots +\frac{(1-s)^{n-1}}{n-1} x_n^{n-1}=1~.
    \]
    
    Now, we analyze the asymptotic behavior of $x_n$.
    Obviously, we have $x_n= \bigO(\tfrac{1}{n-1})$; otherwise, the left-hand side either diverges to infinity or converges to $0$, thus cannot equal $1$.
    
    Now, we find the constant such that $(n-1)(1-s)x_n \to c$ for some constant $c$ as $n \to \infty$.
    Note that such a constant $c$ shall satisfy
    \[
        c + \frac{n-2}{n-1} \cdot \frac{c^2}{2!2} + \cdots + \frac{(n-2)(n-3) \cdots (n-k)}{(n-1)^{k-1}} \cdot \frac{c^k}{k!k} + \cdots = 1.
    \]
    Intuitively, as $n \to \infty$, the coefficients $\frac{(n-2)(n-3)\cdots (n-k)}{(n-1)^{k-1}} \to 1$ for any fixed $k$.
    Thus, we expect that the value $c$ satisfies
    \[
        \sum_{k\ge1} \frac{c^k}{k!k} = 1~,
    \]
    which defines the constant $\constint \approx 0.80435$ in the lemma.
    
    We can prove the convergence more rigorously by showing that, for any $ c \in (0, 1)$, the difference between the ideal and the actual sum tends to $0$ as $n \to \infty$.
    
    Consider the difference:
    \[
        \sum_{k=1}^{n} \frac{c^k}{k!k} \left(1-\left(1-\frac{1}{n-1}\right)\left(1-\frac{2}{n-1}\right)\cdots \left(1-\frac{k-1}{n-1}\right)\right)~.
    \]
    Observe that $1 \leq (1-x_1)(1-x_2)\cdots (1-x_m) + x_1 + x_2 + \cdots + x_m$ for any $x_i\in [0,1]$.
    Indeed, the partial derivative shows that the right-hand side is increasing in each $x_i$, thus the minimum is attained at $x_i=0$ for all $i$, which is $1$.
    
    It follows that the difference is bounded by
    \[
        \sum_{k=1}^{n} \frac{c^k}{k!k} \cdot \frac{1+2+\cdots+(k-1)}{n-1} 
        = \sum_{k=1}^{n} \frac{c^k}{k!k} \cdot \frac{k(k-1)}{2(n-1)}
        \leq \frac{1}{2(n-1)} \sum_{k = 1}^{\infty} \frac{1}{(k-1)!}~.
    \]
    Using the fact that $\sum_{k = 1}^{\infty} \frac{1}{(k-1)!} = \ee$, the bound is $\bigO(\tfrac{1}{n})$ and therefore tends to $0$ as $n \to \infty$.
    
    Finally, we conclude that $(n-1)(1-s)x_n \to \constint$ as $n\to \infty$, which implies that $x_n = \frac{\constint}{(n-1)(1-s)} + o(\tfrac{1}{n})$. This completes the proof by $\maxproboptthres(s) = \frac{1}{1+x_n}=\frac{1}{1+\frac{\constint}{(n-1)(1-s)}} +o(\tfrac{1}{n})$.
\end{proof}

\section{Missing Proofs for \MaxProb Hardness}
\label{sec:app-hard-max-prob}

\subsection{Minor-Oblivious Algorithms: Proof of \Cref{prop:minor-oblivious}}
\label{app:proof-minor-oblivious}

The key observation is that, when the current value is a record and equal across sequences, the conditional probability of stopping at time $t$ is invariant across all truncation levels $k \geq \ell$ for which the conditioning event has positive probability. This is formalized in the following lemma.

\begin{lemma}%
\label{lem:stopping-prob-invariance}
    For any $t \in [n]$ and any $\ell \in [K]$, the conditional probabilities
    \begin{equation*}
        \Pr_k \big[ \,\stoptime = t \mid \stoptime \geq t,\ x_{t} = \max \bm{x}_{1:t} = \ell \,\big] 
    \end{equation*}
    are identical for all $k$ with $\ell \le k \le K$.
\end{lemma}

\begin{proof}
    By Bayes' rule, we can obtain that
    \[
        \Pr_k \big[ \,\stoptime = t \mid \stoptime \geq t \, , \, x_{t} = \max \bm{x}_{1:t} = \ell \,\big] 
        ~=~ 
        \frac{\Pr_k \big[ \,\stoptime = t \, , \, x_{t} = \ell \mid \max \bm{x}_{1:t} = \ell \,\big]}{\Pr_k \big[ \,\stoptime \geq t \, , \, x_{t} = \ell \mid \max \bm{x}_{1:t} = \ell \,\big]}
        ~.
    \]

    Denote by $S \subseteq [K]^t$ the set of all $\bm{x}_{1:t}$ with $x_t = \max \bm{x}_{1:t} = \ell$.
    The numerator can be written as
    \[
        \sum_{\bm{x}_{1:t} \in S} \Pr \big[\, \stoptime = t \mid \bm{x}_{1:t} \,\big] \cdot \Pr_k \big[\, \bm{x}_{1:t} \mid \max \bm{x}_{1:t} = \ell \,\big]
        ~,
    \]
    where the second factor is the conditional probability of realizing $\bm{x}_{1:t}$ given $\max \bm{x}_{1:t} = \ell$.
    This is invariant for all $k$ with $\ell \le k \le K$ by \Cref{lem:hardness-posterior}.
    The same holds for the denominator.
    Hence, the ratio is invariant.
\end{proof}

Given any algorithm $\mathcal{A}$, whose index of accepted value is $\stoptime$, we construct a minor-oblivious algorithm $\mathcal{A}'$ with conditional acceptance probabilities
\[
    r'_t(\ell) ~\coloneqq~ \Pr_K \big[ \,\stoptime = t \mid \stoptime \geq t,\ x_{t} = \max \bm{x}_{1:t} = \ell \,\big] 
    ~.
\]
That is, the probability of accepting a record $x_t = \ell$ equals the probability that $\mathcal{A}$ accepts at time $t$, conditioned on 1) $x_t = \ell$; 2) $x_t$ is the record; and 3) $\mathcal{A}$ rejects the first $t-1$ values.
By \Cref{lem:stopping-prob-invariance}, we may define this conditional probability with respect to prior $F_K$ without loss of generality.

Finally, we show that algorithm $\mathcal{A}'$ accepts the maximum value with the same probability as $\mathcal{A}$.
Let $\stoptime'$ be the index of the accepted value by algorithm $\mathcal{A}'$.
For any prior $F_k$, observe that
\begin{align*}
    \Pr_k \left[x_{\stoptime} = \OPT \right] & 
    ~=~
    \sum_{t \in [n]} \sum_{\ell\in [k]} 
    \, \Pr_k \left[\max \bm{x}_{t+1:n} \leq \ell\right] \cdot \Pr\nolimits_{k}\left[ \stoptime \geq t \,,\, x_{t} = \max \bm{x}_{1:t} = \ell \right] \\
    & \qquad\qquad\qquad 
    \cdot \Pr\left[\stoptime = t \mid \stoptime \geq t \,,\, x_{t} = \max \bm{x}_{1:t} = \ell \right] 
    ~.
\end{align*}

Similarly, we have
\begin{align*}
    \Pr_k \left[x_{\stoptime'} = \OPT \right] & 
    ~=~
    \sum_{t \in [n]} \sum_{\ell\in [k]} 
    \, \Pr_k \left[\max \bm{x}_{t+1:n} \leq \ell\right] \cdot \Pr\nolimits_{k}\left[ \stoptime' \geq t \,,\, x_{t} = \max \bm{x}_{1:t} = \ell \right] \\
    & \qquad\qquad\qquad 
    \cdot \Pr\left[\stoptime' = t \mid \stoptime' \geq t \,,\, x_{t} = \max \bm{x}_{1:t} = \ell \right] 
    ~.
\end{align*}

Comparing the two expressions, the first parts coincide, and the third parts are equal by the construction of $\mathcal{A}'$.
It remains to prove that for any $t \in [n]$ and any $k, \ell \in [K]$:
\begin{equation}
    \label{eq:equal-arrival-prob-T}
    \Pr_k \big[\, \stoptime' \geq t \, , \, x_{t} = \max \bm{x}_{1:t} = \ell \,\big] 
    ~=~ 
    \Pr_k \big[\, \stoptime \geq t \, , \, x_{t} = \max \bm{x}_{1:t} = \ell \,\big]
    ~.
\end{equation}

We next prove this by a joint induction with the following claim:
\begin{equation}
    \label{eq:equal-arrival-prob-F}
    \Pr_k \big[\, \stoptime' \geq t \, , \, x_{t} < \max \bm{x}_{1:t} = \ell \,\big] 
    ~=~ 
    \Pr_k \big[\, \stoptime \geq t \, , \, x_{t} < \max \bm{x}_{1:t} = \ell \,\big]
    ~.
\end{equation}

The base case of $t=1$ holds trivially, as the probabilities are independent of the algorithms' decisions.
Next, suppose that both equalities hold for some $t \ge 1$. 
We will prove them for $t+1$.
Enumerating over the state at step $t$, we have:
\begin{align}
    &
    \Pr_k \big[\, \stoptime \geq t+1 \,,\, x_{t+1} = \max \bm{x}_{1:t+1} = \ell \,\big] \notag \\[2ex]
    & \quad
    ~=~
    \sum_{m \in [\ell]} \Pr_k \big[ \,\stoptime \geq t \, , \, x_{t} = \max \bm{x}_{1:t} = m \,\big] \cdot \Pr_k \big[x_{t+1} = \ell \big] \label{eq:minor-oblivious-induction-1}\\
    & \qquad\qquad
    \cdot\, \big(1 -\Pr[\stoptime = t \mid \stoptime \geq t \, , \, x_{t} = \max \bm{x}_{1:t} = m ]\big)  \label{eq:minor-oblivious-induction-2}\\
    & \quad\qquad
    + \sum_{m \in [\ell]} \Pr\nolimits_{k}\big[\stoptime \geq t \, , \, x_{t} < \max \bm{x}_{1:t} = m\big] \cdot \Pr_k \big[x_{t+1} = \ell \big] \, .
    \label{eq:minor-oblivious-induction-3}
\end{align}

By the induction hypothesis, we can replace $\stoptime$ with $\stoptime'$ in \eqref{eq:minor-oblivious-induction-1} and \eqref{eq:minor-oblivious-induction-3} while preserving the value.
Further, we can also replace $\stoptime$ with $\stoptime'$ in \eqref{eq:minor-oblivious-induction-2} by the construction of algorithm $\mathcal{A}'$.
Hence, \Cref{eq:equal-arrival-prob-T} holds for $t+1$ as well.

The proof of \Cref{eq:equal-arrival-prob-F} for $t+1$ is almost verbatim.

\subsection{Relating Rejection and Acceptance Probabilities: Proof of \Cref{lem:survival-to-stopping}}
\label{app:survival-stopping}

\paragraph{From Rejection Probabilities to Acceptance Probabilities.}  
By definition, we have:
\begin{align*}
    \accept_t(\ell) 
    ~=~& \Pr \big[\, \stoptime = t \mid \stoptime \geq t \, , \, x_t = \max \bm{x}_{1:t}=\ell \,\big] \\
    ~=~& 1 - \Pr\big[\, \stoptime > t \mid \stoptime \geq t \, , \, x_t = \max \bm{x}_{1:t} = \ell \, \big] \, .
\end{align*}

By Bayes' rule, the probability on the right-hand side can be written as
\[
\Pr\big[\, \stoptime > t \mid \stoptime \geq t \, , \, x_t = \max \bm{x}_{1:t} = \ell \, \big] 
~=~
\frac{\Pr_K \big[\, \stoptime > t \, , \, x_t = \max \bm{x}_{1:t} = \ell \,\big]}{\Pr_K\big[\, \stoptime \geq t \, , \, x_t = \max \mathbf{x}_{1:t} = \ell \,\big]}
\, .
\]
Here we consider the prior $F_K$ without loss of generality, because using any other $F_k$ multiplies numerator and denominator by the same $\tilde{F}(k)^{-t}$ factor according to \Cref{def:truncated-distribution-family}.

The numerator further equals
\begin{align*}
	& \Pr_K \big[\, \stoptime > t \,,\, \max \bm{x}_{1:t} = \ell \,\big] ~-~ \Pr_K \big[\, \stoptime > t \,,\, x_t < \max \bm{x}_{1:t} = \ell \,\big] \\
	& \quad
	~=~ \Pr_K \big[\, \max \bm{x}_{1:t} = \ell \,\big] \cdot \rejall_t(\ell) - \Pr_K \big[\, x_t < \max \bm{x}_{1:t-1} = \ell \,\big] \cdot \rejall_{t-1}(\ell)~.
\end{align*}
The denominator equals
\[
    \sum_{m \in [\ell]} \Pr_K \big[\, \stoptime \ge t \, , \, \max \bm{x}_{1:t-1} = m \, , \, x_t = \ell \,\big]
    = 
    \sum_{m \in [\ell]} \Pr_K \big[\, \max \bm{x}_{1:t-1} = m \, , \, x_t = \ell \,\big] \cdot \rejall_{t-1}(m)
    \, .
\]

Finally, putting these together proves \Cref{lem:survival-to-stopping}.

\paragraph{From Acceptance Probabilities to Rejection Probabilities.}  
We remark that the relation between acceptance probabilities and rejection probabilities is one-to-one. Given the acceptance probabilities, we can inductively deduce the rejection probabilities using \Cref{lem:survival-to-stopping}. Restating \Cref{eq:Rt-def} from \Cref{lem:survival-to-stopping}:
\[
    \accept_t(\ell) ~=~ 1 - \frac{\rejall_t(\ell) \cdot \Pr_K \big[ \max \bm{x}_{1:t} = \ell \,\big] - \rejall_{t-1}(\ell) \cdot \Pr_K \big[\, x_t < \max \bm{x}_{1:t-1} = \ell \,\big]}
    {\sum_{m \in [\ell]} \rejall_{t-1}(m) \cdot \Pr_K \big[ \max \bm{x}_{1:t-1} = m, x_t = \ell \,\big]}
    \, .
\]
Given $\rejall_{t-1}(\cdot)$ and $\accept_t(\cdot)$ for any $t \ge 1$, we can derive $\rejall_t(\cdot)$ from the equation above.

\end{document}